%% file: Domination-circle.tex
\documentclass[runningheads,11points]{llncs}

\usepackage{amssymb,hyperref}
\usepackage{cite,scrtime}
\usepackage{graphics}
\usepackage{amsmath,dsfont}
\usepackage{mathrsfs}
\usepackage{fancyhdr}
\usepackage{verbatim}
\usepackage{boxedminipage}
\usepackage{colortbl}
\usepackage{algorithm}
\usepackage[noend]{algorithmic}

\usepackage{graphicx}

\usepackage{vmargin}
\setmarginsrb{3.8cm}{1.2cm}{3.8cm}{3.cm}{1.2cm}{1.2cm}{0.8cm}{1.4cm}

\def\debut{\begin{itemize}\item[{\bf [[}]\small}
\def\term{\hfill {\bf]]} \end{itemize} }

\newtheorem{claimN}{Claim}

\renewenvironment{proof}[1][]{\par \noindent {\bf Proof:#1}\ }{\hfill$\Box$\\}

\title{Parameterized Domination in Circle Graphs\thanks{A preliminary conference version
of this work appeared in the \emph{Proceedings of the 38th International
Workshop on Graph-Theoretic Concepts in Computer Science (WG), Jerusalem,
Israel, June 2012}. The third author was partially supported by EPSRC Grant
EP/G043434/1. The other authors were partially supported by AGAPE
(ANR-09-BLAN-0159) and GRATOS (ANR-09-JCJC-0041) projects (France).}}
\author{Nicolas Bousquet~\inst{1}, Daniel Gon\c{c}alves~\inst{1}, George B. Mertzios~\inst{2},\\Christophe Paul~\inst{1}, Ignasi Sau~\inst{1}, and St\'ephan Thomass\'e~\inst{3}}

\authorrunning{N. Bousquet, D. Gon\c{c}alves, G. B. Mertzios, C. Paul, I. Sau, and S. Thomass\'e}
\titlerunning{Parameterized Domination in Circle Graphs}

\institute{AlGCo project-team, CNRS, LIRMM, Montpellier, France.\\
\email{\{FirstName.FamilyName\}@lirmm.fr}\\
\and School of Engineering and Computing Sciences, Durham University, United
Kingdom.\\
\email{george.mertzios@durham.ac.uk}\\
\and Laboratoire LIP (U. Lyon, CNRS, ENS Lyon, INRIA, UCBL), Lyon,
France.\\
\email{stephan.thomasse@ens-lyon.fr}}

\begin{document}

\maketitle

\vspace{-.15cm}
\begin{abstract}
A \emph{circle graph} is the intersection graph of a set of chords in a circle.
Keil~[\emph{Discrete Applied Mathematics}, 42(1):51-63, 1993] proved that
\textsc{Dominating Set}, \textsc{Connected Dominating Set}, and \textsc{Total
Dominating Set} are \textsc{NP}-complete in circle graphs. To the best of our
knowledge, nothing was known about the parameterized complexity of these
problems in circle graphs. In this paper we prove the following results, which
contribute in this direction:
\begin{itemize}
\item[$\bullet$] \textsc{Dominating Set}, \textsc{Independent Dominating Set},
\textsc{Connected Dominating Set}, \textsc{Total Dominating Set}, and
\textsc{Acyclic Dominating Set} are $W[1]$-hard in circle graphs, parameterized
by the size of the solution.
\item[$\bullet$] Whereas both \textsc{Connected Dominating Set} and \textsc{Acyclic
Dominating Set} are $W[1]$-hard in circle graphs, it turns out that
\textsc{Connected Acyclic Dominating Set} is polynomial-time solvable in circle
graphs.
\item[$\bullet$] If $T$ is a \emph{given} tree, deciding whether a
circle graph has a dominating set isomorphic to $T$ is \textsc{NP}-complete
when $T$ is in the input, and \textsc{FPT} when parameterized by $|V(T)|$. We
prove that the FPT algorithm is subexponential.
\end{itemize}
\vspace{0.25cm} \textbf{Keywords:} circle graphs; domination problems;
parameterized complexity; parameterized algorithms; dynamic programming;
constrained domination.
\end{abstract}

\section{Introduction}
\label{sec:intro}

A \emph{circle graph} is the intersection graph of a set of chords in a circle
(see Fig.~\ref{fig:example} for an example of a circle graph $G$ together with
a circle representation of it). The class of circle graphs has been extensively
studied in the literature, due in part to its applications to
sorting~\cite{EvIt71} and VLSI design~\cite{She92}. Many problems which are
\textsc{NP}-hard in general graphs turn out to be solvable in polynomial time
when restricted to circle graphs. For instance, this is the case of
\textsc{Maximum Clique} and \textsc{Maximum Independent Set}~\cite{Gav73},
\textsc{Treewidth}~\cite{Klo96}, \textsc{Minimum Feedback Vertex
Set}~\cite{Gav08}, \textsc{Recognition}~\cite{GPTC11,Spi94}, \textsc{Dominating
Clique}~\cite{Kei93}, or \textsc{$3$-Colorability}~\cite{Ung92}.

\begin{figure}[tb]
\center \vspace{-.4cm}
\includegraphics[width=0.7\textwidth]{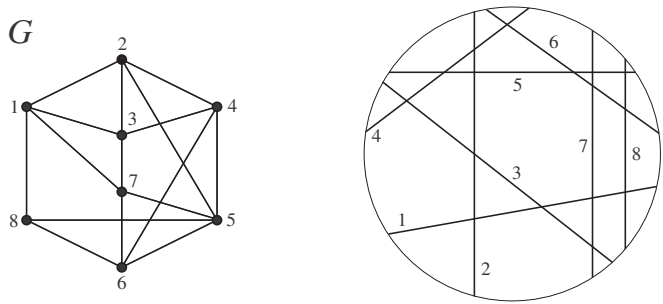}
\vspace{-.2cm} \caption{A circle graph $G$ on $8$ vertices together with a
circle representation of it.\label{fig:example}}\vspace{-.3cm}
\end{figure}

But still a few problems remain \textsc{NP}-complete in circle graphs, like
\textsc{$k$-Colorability} for $k \geq 4$~\cite{Ung88}, \textsc{Hamiltonian
Cycle}~\cite{Dam89}, or \textsc{Minimum Clique Cover}~\cite{KeSt06}. In this
article we study a variety of domination problems in circle graphs, from a
parameterized complexity perspective. A \emph{dominating set} in a graph
$G=(V,E)$ is a subset $S \subseteq V$ such that every vertex in $V \setminus S$
has at least one neighbor in $S$. Some extra conditions can be imposed to a
dominating set. For instance, if $S \subseteq V$ is a dominating set and $G[S]$
is connected (resp. acyclic, an independent set, a graph without isolated
vertices, a tree, a path), then $S$ is called a \emph{connected} (resp.
\emph{acyclic}, \emph{independent}, \emph{total}, \emph{tree}, \emph{path})
\emph{dominating set}. In the example of Fig.~\ref{fig:example}, vertices $1$
and $5$ (resp. $3$, $4$, and $6$) induce an independent (resp. connected)
dominating set. The corresponding minimization problems are defined in the
natural way. Given a set of graphs $\mathcal{G}$, the \textsc{Minimum
$\mathcal{G}$-Dominating Set} problem consists in, given a graph $G$, finding a
dominating set $S \subseteq V(G)$ of $G$ of minimum cardinality such that
$G[S]$ is isomorphic to some graph in $\mathcal{G}$. Throughout the article, we
may omit the word ``\textsc{Minimum}'' when referring to a specific problem.

For an introduction to parameterized complexity theory, see for
instance~\cite{DoFe99,Nie06,FlGr06}. A decision problem with input size $n$ and
parameter $k$ having an algorithm which solves it in time $f(k)\cdot
n^{\mathcal{O}(1)}$ (for some computable function $f$ depending only on $k$) is
called \emph{fixed-parameter tractable}, or \textsc{FPT} for short. The
parameterized problems which are $W[i]$-hard for some $i \geq 1$ are not likely
to be \textsc{FPT}~\cite{DoFe99,Nie06,FlGr06}. A parameterized problem is in
\textsc{XP} if it can be solved in time $f(k) \cdot n^{g(k)}$, for some
(unrestricted) functions $f$ and $g$. The parameterized versions of the above
domination problems when parameterized by the cardinality of a solution are
also defined naturally.


\paragraph{\textbf{\emph{Previous work.}}} \textsc{Dominating Set}  is one of the
most prominent classical graph-theoretic \textsc{NP}-complete
problems~\cite{GaJo79}, and has been studied intensively in the literature.
Keil~\cite{Kei93} proved that \textsc{Dominating Set}, \textsc{Connected
Dominating Set}, and \textsc{Total Dominating Set} are \textsc{NP}-complete
when restricted to circle graphs, and Damian and Pemmaraju~\cite{DaPe99} proved
that \textsc{Independent Dominating Set} is also \textsc{NP}-complete in circle
graphs, answering an open question from Keil~\cite{Kei93}.

Hedetniemi, Hedetniemi, and Rall~\cite{HHR00} introduced acyclic domination in
graphs. In particular, they proved that \textsc{Acyclic Dominating Set} can be
solved in polynomial time in interval graphs and proper circular-arc graphs.
Xu, Kang, and Shan~\cite{XKS06} proved that \textsc{Acyclic Dominating Set} is
linear-time solvable in bipartite permutation graphs. The complexity status of
\textsc{Acyclic Dominating Set} in circle graphs was unknown.

In the theory of parameterized complexity~\cite{FlGr06,Nie06,DoFe99},
\textsc{Dominating Set} also plays a fundamental role, being the paradigm of a
$W[2]$-hard problem. For some graph classes, like planar graphs,
\textsc{Dominating Set} remains \textsc{NP}-complete~\cite{GaJo79} but becomes
\textsc{FPT} when parameterized by the size of the solution~\cite{ABF+02}.
Other more recent examples can be found in $H$-minor-free graphs~\cite{AlGu08}
and claw-free graphs~\cite{CPP+11}.

The parameterized complexity of domination problems has been also studied in
geometric graphs, like $k$-polygon graphs~\cite{ElSt93}, multiple-interval
graphs and their complements~\cite{FHRV09,ZhJi11}, $k$-gap interval
graphs~\cite{FGG+12}, or graphs defined by the intersection of unit squares,
unit disks, or line segments~\cite{Mar06}. But to the best of our knowledge,
the parameterized complexity of the aforementioned domination problems in
circle graphs was open.

\paragraph{\textbf{\emph{Our contribution.}}} In this paper we prove the
following results, which settle the parameterized complexity of a number of
domination problems in circle graphs:
\begin{itemize}
\item[$\bullet$] In Section~\ref{sec:hard}, we prove that \textsc{Dominating Set}, \textsc{Connected Dominating Set},
\textsc{Total Dominating Set}, \textsc{Independent Dominating Set}, and
\textsc{Acyclic Dominating Set} are $W[1]$-hard in circle graphs, parameterized
by the size of the solution. Note that \textsc{Acyclic Dominating Set} was not
even known to be \textsc{NP}-hard in circle graphs. The reductions are from
$k$-\textsc{Colored Clique} in general graphs.
\item[$\bullet$] Whereas both \textsc{Connected Dominating Set} and \textsc{Acyclic Dominating Set} are $W[1]$-hard in circle graphs,
it turns out that \textsc{Connected Acyclic Dominating Set} is polynomial-time
solvable in circle graphs. This is proved in Section~\ref{sec:algoTrees}.
\item[$\bullet$] Furthermore, if $T$ is a \emph{given} tree, we prove
that the problem of deciding whether a circle graph has a dominating set
isomorphic to $T$ is \textsc{NP}-complete (Section~\ref{sec:hardTrees}) but
\textsc{FPT} when parameterized by $|V(T)|$ (Section~\ref{sec:FPTalgoTrees}).
The $\textsc{NP}$-completeness reduction is from 3-\textsc{Partition}, and we
prove that the running time of the \textsc{FPT} algorithm is subexponential. As
a corollary of the algorithm presented in Section~\ref{sec:FPTalgoTrees}, we
also deduce that, if $T$ has bounded degree, then deciding whether a circle
graph has a dominating set isomorphic to $T$ can be solved in polynomial time.
\end{itemize}

\paragraph{\textbf{\emph{Further research.}}} Some interesting questions remain open. We proved
that several domination problems are $W[1]$-hard in circle graphs. Are they
$W[1]$-complete, or may they also be $W[2]$-hard? On the other hand, we proved
that finding a dominating set isomorphic to a tree can be done in polynomial
time. It could be interesting to generalize this result to dominating sets
isomorphic to a connected graph of fixed treewidth. Finally, even if
\textsc{Dominating Set} parameterized by treewidth is \textsc{FPT} in general
graphs due to Courcelle's theorem~\cite{Cou88}, it is not plausible that it has
a polynomial kernel in general graphs~\cite{BDFH09}. It may be the case that
the problem admits a polynomial kernel parameterized by treewidth (or by vertex
cover) when restricted to circle graphs.

\section{Hardness results}
\label{sec:hard}

In this section we prove hardness results for a number of domination problems
in circle graphs. In order to prove the $W[1]$-hardness of the domination
problems, we provide two families of reductions. Namely, in
Section~\ref{sec:hardness1} we prove the hardness of \textsc{Dominating Set},
\textsc{Connected Dominating Set}, and \textsc{Total Dominating Set}, and in
Section~\ref{sec:hardness2} we prove the hardness of \textsc{Independent
Dominating Set} and \textsc{Acyclic Dominating Set}. Finally, we prove the
\textsc{NP}-completeness for trees in Section~\ref{sec:hardTrees}.

For better visibility, some figures of this section have colors, but these
colors are not indispensable for completely understanding the depicted
constructions. Before stating the hardness results, we need to introduce the
following parameterized problem, proved to be $W[1]$-hard in~\cite{FHRV09}.

\begin{center}
\begin{boxedminipage}{.8\textwidth}

$k$-\textsc{Colored Clique}

\begin{tabular}{ r l }

\textit{~~~~Instance:} & A graph $G=(V,E)$ and a coloring of $V$ using $k$ colors.\\
\textit{Parameter:} & $k$.\\
\textit{Question:} & Does there exist a clique of size $k$ in $G$ containing\\
& exactly one vertex from each color?
\end{tabular}
\end{boxedminipage}
\end{center}

Note that in an instance of $k$-\textsc{Colored Clique}, we can assume that
there is no edge between any pair of vertices colored with the same color.
Also, we can assume that for each $1 \leq i \leq k$, the number of vertices
colored with color $i$ is the same. Indeed, given an instance $G$, we can
consider an equivalent instance $G'$ obtained by putting together $k!$ disjoint
copies of $G$, one for each permutation of the color classes.


In a representation of a circle graph, we will always consider the circle
oriented anticlockwise. Given three points $a,b,c$ in the circle, by $a < b <
c$ we mean that starting from $a$ and moving anticlockwise along the circle,
$b$ comes before $c$. In a circle representation, we say that two chords with
endpoints $(a,b)$ and $(c,d)$ are \emph{parallel twins} if $a < c < d < b$, and
there is no other endpoint of a chord between $a$ and $c$, nor between $d$ and
$b$. Note that for any pair of parallel twins $(a,b)$ and $(c,d)$, we can slide
$c$ (resp. $d$) arbitrarily close to $a$ (resp. $b$) without modifying the
circle representation.

\subsection{Hardness of domination, and connected and total domination}
\label{sec:hardness1}

We start with the main result of this section.

\begin{theorem}\label{dominationw1}
\textsc{Dominating Set} is $W[1]$-hard in circle graphs, when parameterized by
the size of the solution.
\end{theorem}
\begin{proof}
We shall reduce the $k$-\textsc{Colored Clique} problem to the problem of
finding a dominating set of size at most $k(k+1)/2$ in circle graphs. Let $k$
be an integer and let $G$ be a $k$-colored graph on $kn$ vertices such that $n$
vertices are colored with color $i$ for all $1 \leq i \leq k$. For every $1
\leq i \leq k$, we denote by $x_j^i$ the vertices of color $i$, with $1 \leq j
\leq n$. Let us prove that $G$ has a $k$-colored clique of size $k$ if and only
if the following circle graph $C$ has a dominating set of size at most
$k(k+1)/2$. We choose an arbitrary point of the circle as the \emph{origin}.
The circle graph $C$ is defined as follows:


\begin{figure}[t]
\center \vspace{-.55cm}
\includegraphics[width=0.56\textwidth]{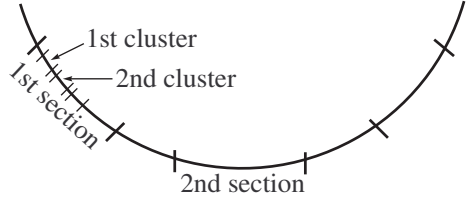}
\vspace{-.4cm} \caption{Sections and clusters in the reduction of
Theorem~\ref{dominationw1}.} \label{example}
\end{figure}

\begin{itemize}
\item[$\bullet$] We divide the circle into $k$ disjoint open intervals
$]s_i,s'_i[$ for $1 \leq i \leq k$, called \emph{sections}. Each section is
divided into $k+1$ disjoint intervals $]c_{ij},c'_{ij}[$ for $1 \leq j \leq
k+1$, called \emph{clusters}
 (see Fig.~\ref{example} for an illustration). Each cluster has $n$
particular points denoted by $1,\ldots,n$ following the order of the circle.
These intervals are constructed in such a way that the origin is not in a
section.

\item[$\bullet$] Sections are numbered from $1$ to $k$ following the
anticlockwise order from the origin. Similarly, the clusters inside each
section are numbered from $1$ to $k+1$.


\item[$\bullet$] For each $1 \leq i \leq k, 1 \leq j \leq k+1$, we
add a chord with endpoints $c_{ij}$ and $c_{ij}'$, which we call the
\emph{extremal chord} of the $j$-cluster of the $i$-th section.

\item[$\bullet$] For each $1 \leq i \leq k$ and $1 \leq j \leq k$, we add
chords between the $j$-th and the $(j+1)$-th clusters of the $i$-th section as
follows. For each $0 \leq l \leq n$, we add two parallel twin chords, each
having one endpoint in the interval $]l,l+1[$ of the $j$-th cluster, and the
other endpoint in the interval $]l,l+1[$ of the $(j+1)$-th cluster. These
chords are called \emph{inner chords} (see Fig.~\ref{inner} for an
illustration). We note that the endpoints of the inner chords inside each
interval can be chosen arbitrarily. The interval $]0,1[$ is the interval
between $c_{ij}$ and the point $1$, and similarly $]n,n+1[$ is the interval
between the point $n$ and $c_{ij}'$.

\item[$\bullet$] We also add chords between the first and the last clusters of
each section. For each $1 \leq i \leq k$ and $1 \leq l \leq n$, we add a chord
joining the point $l$ of the first cluster and the point $l$ of the last
cluster of the $i$-th section. For each $1 \leq i \leq k$, these chords are
called the \emph{$i$-th memory chords}. 

\item[$\bullet$] Extremal, inner, and memory chords will ensure some structure on
the solution. On the other hand, the following chords will simulate the
behavior of the original graph. In fact, the $n$ particular points in each
cluster of the $i$-th section will simulate the behavior of the $n$ vertices of
color $i$ in $G$. Let $i < j$. The chords from the $i$-th section to the $j$-th
section are between the $j$-th cluster of the $i$-th section and the $(i+1)$-th
cluster of the $j$-th section. Between this pair of clusters, we add a chord
joining the point $h$ (in the $i$-th section) and the point $l$ (in the $j$-th
section) if and only if $x^i_hx^j_l \in E(G)$. We say that such a chord is
called \emph{associated} with an edge of the graph $G$, and such chords are
called \emph{outer chords}. In other words, there is an outer chord in $C$ if
the corresponding vertices are connected in $G$.\end{itemize}
\begin{figure}
\center \vspace{-.45cm}
\input{innerchords.pstex_t}
\caption{Representation of the chords between the $j$-th and the $(j+1)$-th
cluster of the $i$-th section. The higher chords are extremal chords. The
others are inner chords and have to be replaced by two parallel twin chords.}
\label{inner}
\end{figure}
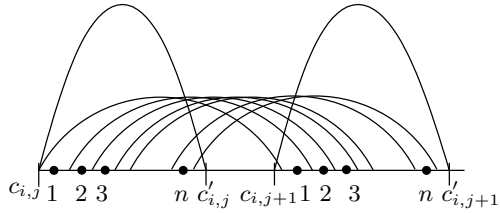

Intuitively, the idea of the above construction is as follows. For each $1 \leq
i \leq k$, among the $k+1$ clusters in the $i$-th section, the first and the
last one do not contain endpoints of outer chords, and are only used for
technical reasons (as discussed below). The remaining $k-1$ clusters in the
$i$-th section capture the edges of $G$ between vertices of color $i$ and
vertices of the remaining $k-1$ colors. Namely, for any two distinct colors $i$
and $j$, there is a cluster in the $i$-th section and a cluster in the $j$-th
section such that the outer chords between these two clusters correspond to the
edges in $G$ between colors $i$ and $j$. The rest of the proof is structured
along a series of claims.

\begin{claimN}
\label{1domw1} If there exists a $k$-colored clique in $G$, then there exists a
dominating set of size $k(k+1)/2$ in $C$.
\end{claimN}
\begin{proof}
Assume that there is a $k$-colored clique $K$ in $G$ and let us denote by $k_i$
the integer such that $x_{k_i}^i$ is the vertex of color $i$ in this clique.
Let $\mathcal{D}$ be the following set of chords. For each section $1 \leq i
\leq k$, we add to $\mathcal{D}$ the memory chord joining the points $k_i$ of
the first and the last clusters. We also add in $\mathcal{D}$ the outer chords
associated with the edges of the $k$-colored clique. The set $\mathcal{D}$
contains $k(k+1)/2$ chords: $k$ memory chords and $k(k-1)/2$ outer chords. Let
us prove that $\mathcal{D}$ is a dominating set.

The extremal chords are dominated, since $\mathcal{D}$ has exactly one endpoint
in each cluster. Indeed, there is an endpoint in the first and the last cluster
of section because of the memory chords of $\mathcal{D}$. There is an endpoint
in the other clusters because of the outer chord associated with the edge of
the $k$-colored clique. The inner chords are also dominated. Indeed, for each
section $i$, the endpoint of the chord of $\mathcal{D}$ is $k_i$, for all
clusters $j$ such that $1 \leq j \leq k+1$. Thus, for all $1 \leq l \leq n$,
the inner chords between the intervals $]l,l+1[$ of the $j$-th cluster and
$]l,l+1[$ of the $(j+1)$-th cluster are dominated by the chord of $\mathcal{D}$
with endpoint in the $j$-th section if $l \leq k_i-1$, or by the chord of the
$(j+1)$-th section otherwise.

The outer chords are dominated by the memory chords of $\mathcal{D}$, since the
outer chords have their endpoints in two different sections. Finally, the
memory chords are also dominated by the outer chords of $\mathcal{D}$ for the
same reason.
\end{proof}

In the following we will state some properties about the dominating sets in $C$
of size $k(k+1)/2$.
\begin{claimN}
\label{oneincluster} A dominating set in $C$ has size at least $k(k+1)/2$, and
a dominating set of this size has exactly one endpoint in each cluster.
\end{claimN}
\begin{proof}
The \emph{interval of a chord} linking $x$ and $y$, with $x < y$, is the
interval $[x,y]$. One can note that when $\ell$ chords have pairwise disjoint
intervals, at least $\lceil \ell/2 \rceil$ chords are necessary to dominate
them.


The intervals $[c_{i,j},c'_{i,j}]$ of the extremal chords are disjoint by
assumption. Since there are $k(k+1)$ extremal chords, a dominating set has size
at least $k(k+1)/2$. And if there is a dominating set of such size, it must
have exactly one endpoint in each interval, i.e., one endpoint in each cluster.
\end{proof}

\begin{claimN}\label{exactlyone}
A dominating set of size $k(k+1)/2$ in $C$ contains no inner nor extremal
chord.
\end{claimN}
\begin{proof}
Let $\mathcal{D}$ be a dominating set in $C$ of size $k(k+1)/2$. It contains no
extremal chords, since both endpoints of an extremal chord are in the same
cluster, which is impossible by Claim~\ref{oneincluster}. If $\mathcal{D}$
contains an inner chord $c$, the parallel twin of $c$ in $C$ is dominated by
some other chord $c'$. But then $c \cup c'$ intersect at most three clusters,
which is again impossible by Claim~\ref{oneincluster}.
\end{proof}

By Claim~\ref{exactlyone}, a dominating set in $C$ of size $k(k+1)/2$ contains
only memory and outer chords. Thus, the unique (by Claim~\ref{oneincluster})
endpoint of the dominating set in each cluster is one of the points
$\{1,\ldots,n\}$, and we call it the \emph{value} of a cluster.
Fig.~\ref{formsolution} illustrates the general form of a solution.

\begin{figure}[t]
\center \vspace{-.3cm}
\includegraphics[scale=0.5]{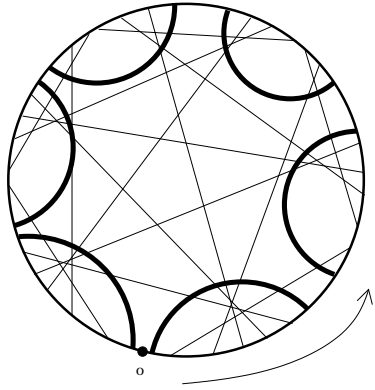}
\caption{The general form of a solution in the reduction of
Theorem~\ref{dominationw1}. The thick chords are memory chords and the other
ones are outer chords. The origin is depicted with a small ``o''.}
\label{formsolution}
\end{figure}

\begin{claimN}\label{selection}
Assume that $C$ contains a dominating set of size $k(k+1)/2$. Then, in a given
section, the value of a cluster does not increase between consecutive clusters.
\end{claimN}
\begin{proof}
Assume that in a given arbitrary section, the value of the $j$-th cluster is
$l$. The inner chords between the interval $]l,l+1[$ of the $j$-th cluster and
the interval $]l,l+1[$ of the $(j+1)$-th cluster have to be dominated. Since
the value of the $j$-th cluster is $l$, they are not dominated in the $j$-th
cluster. Therefore, in order to ensure the domination of these chords, the
value of the $(j+1)$-th cluster is at most $l$.
\end{proof}

\begin{claimN}
\label{claim:value} Assume that $C$ contains a dominating set of size
$k(k+1)/2$. Then, for each $1 \leq i \leq k$, all the clusters of the $i$-th
section have the same value.
\end{claimN}
\begin{proof}
Let  $\mathcal{D}$ be such a dominating set. In a given section, the endpoints
of $\mathcal{D}$ in the first and the last clusters are endpoints of a memory
chord, and for all $l$, they link the point $l$ of the first cluster to the
point $l$ of the last one. Thus, the first and the last clusters have the same
value. Since by Claim~\ref{selection} the value of a cluster decreases between
consecutive clusters, the value of the clusters of the same section is
necessarily constant.
\end{proof}

The \emph{value} of a section is the value of the clusters in this section
(note that it is well-defined by Claim~\ref{claim:value}). The \emph{vertex
associated with the $i$-th section} is the vertex $x^i_k$ if the value of the
$i$-th section is $k$.

\begin{claimN}\label{2domw1}
If there is a dominating set in $C$ of size $k(k+1)/2$, then for each pair
$(i,j)$ with $1 \leq i < j \leq k$, the vertex associated with the $i$-th
section is adjacent in $G$ to the vertex associated with the $j$-th section.
Therefore, $G$ has a $k$-colored clique.
\end{claimN}
\begin{proof}
Let $i$ and $j$ be two sections with $i<j$, and let $x^i_k$ and $x^j_l$ be the
vertices associated with these two sections, respectively. By
Claim~\ref{claim:value}, the chord of the dominating set in the $j$-th cluster
of the $i$-th section has a well-defined endpoint $k$, and the chord of the
dominating set in the $(i-1)$-th cluster of the $j$-th section has a
well-defined endpoint $l$. The vertex $x_k^i$ associated with the $i$-th
section is adjacent in $G$ to the vertex $x_l^j$ associated with the $j$-th
section. Indeed, the chords having endpoints in these clusters are exactly the
chords between these two clusters, and there is a chord if and only if there is
an edge between the corresponding vertices in $G$.
\end{proof}

Claims~\ref{1domw1} and \ref{2domw1} together ensure that $C$ has a dominating
set of size $k(k+1)/2$ if and only if $G$ has a $k$-colored clique. The
reduction can be easily done in polynomial time, and the parameters of the
problems are polynomially equivalent. Thus, \textsc{Dominating Set} in circle
graphs is $W[1]$-hard. This completes the proof of Theorem~\ref{dominationw1}.
\end{proof}

From Theorem~\ref{dominationw1} we can easily deduce the $W[1]$-hardness of two
other domination problems in circle graphs.

\begin{corollary}\label{connecteddom}
\textsc{Connected Dominating Set} and \textsc{Total Dominating Set} are
$W[1]$-hard in circle graphs, when parameterized by the size of the solution.
\end{corollary}
\begin{proof}
In the construction of Theorem~\ref{dominationw1}, if there is a dominating set
of size $k(k+1)/2$ in $C$, it is necessarily connected (see the form of the
solution in Fig. \ref{formsolution}). Indeed, the memory chords ensure the
connectivity between all the chords with one endpoint in a section. Since there
is a chord between each pair of sections, the dominating set is connected.
Finally, note that a connected dominating set is also a total dominating set,
as it contains no isolated vertices.
\end{proof}

%

\subsection{Hardness of independent and acyclic domination}
\label{sec:hardness2}

We proceed to describe our second construction in order to prove parameterized
reductions for domination problems in circle graphs.

\begin{theorem}
\label{th:independent} \textsc{Independent Dominating Set} is $W[1]$-hard in
circle graphs.
\end{theorem}
\begin{proof} We present a parameterized reduction from $k$-\textsc{Colored
Clique} in a general graph to the problem of finding an independent dominating
set of size at most $2k$ in a circle graph. Let $G$ be the input $k$-colored
graph with color classes $X^1,\ldots,X^k \subseteq V(G)$. Let
$x^{i}_1,\ldots,x^{i}_n$ be the vertices belonging to the color class $X^i
\subseteq V(G)$, in an arbitrary order. We proceed to build a
circle graph $H$ by defining its circle representation. 
Let $I^1,\ldots,I^k$ be a collection of $k$ disjoint intervals in the circle,
which will we associated with the $k$ colors. For $1 \leq i \leq k$, we proceed
to construct an induced subgraph $H^i$ of $H$ whose chords have all endpoints
in the interval $I^i$, which we visit from left to right. Throughout the
construction, cf. Fig.~\ref{fig:independent} for an example with $n=4$.

\begin{figure}[t]
\flushleft \vspace{-1.2cm}
\includegraphics[width=1.0\textwidth]{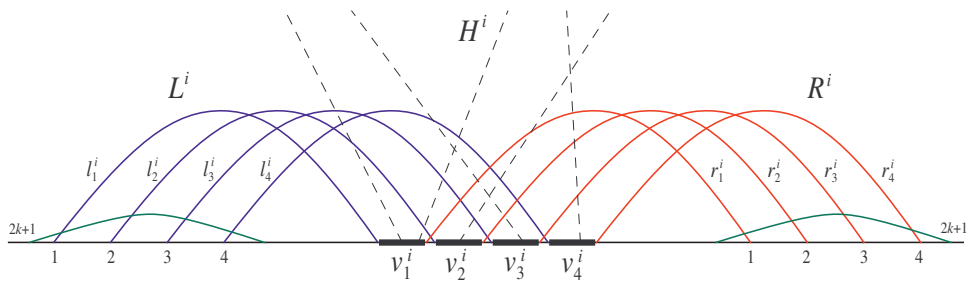}
\vspace{-.6cm} \caption{Gadget $H^i$ in interval $I^i$ used in the proof of
Theorem~\ref{th:independent}, corresponding to a color class $X^i$ of the
$k$-colored input graph $G$. The dashed chords correspond to non-edges of
$G$.\label{fig:independent}}\vspace{-.3cm}
\end{figure}

We start by adding two cliques on $n$ vertices $L^i$ and $R^i$, with chords
$l^{i}_1,\ldots,l^{i}_n$ and $r^{i}_1,\ldots,r^{i}_n$, respectively, in the
following way. The endpoints of $L^i$ and $R^i$ are placed in three disjoint
subintervals of $I^i$, such that the first subinterval contains, in this order,
the left endpoints of $l^{i}_1,\ldots,l^{i}_n$. The second subinterval contains
the right endpoints of $L^i$ and the left endpoints of $R^i$, in the order
$l^{i}_1, r^{i}_1, l^{i}_2, r^{i}_2, \ldots, r^{i}_{n-1},l^{i}_n, r^{i}_n$.
Finally, the third subinterval contains, in this order, the right endpoints of
$r^{i}_1,\ldots,r^{i}_n$. The blue (resp. red) chords in
Fig.~\ref{fig:independent} correspond to $L^i$ (resp. $R^i$).

For $1 \leq j \leq n$, we define the interval $v^i_j$ as the open interval
between the right endpoint of $l^i_j$ and the left endpoint of $r^i_j$; cf. the
thick intervals in Fig.~\ref{fig:independent}. Such an interval $v^i_j$ will
correspond to vertex $x^i_j$ of $G$.




We also add two sets of $2k+1$ parallel twin chords whose left endpoints are
placed exactly before the left (resp. right) endpoint of $l^{i}_1$ (resp.
$r^{i}_1$) and whose right endpoints are placed exactly after the left (resp.
right) endpoint of $l^{i}_n$ (resp. $r^{i}_n$); cf. the green chords in
Fig.~\ref{fig:independent}. We call these chords \emph{parallel chords}. This
completes the construction of $H^i$.

Finally, for each pair of vertices $x^{i}_p,x^{j}_q$ of $G$ such that $i \neq
j$ and $\{x^{i}_p,x^{j}_q\} \notin E(G)$, we add to $H$ a chord $c^{i,j}_{p,q}$
between the interval $v^{i}_p$ in $H^i$ and the interval $v^{j}_q$ in $H^j$;
cf. the dashed chords in Fig.~\ref{fig:independent}. We call these chords
\emph{outer chords}. That is, the outer chords of $H$ correspond to non-edges
of $G$. This completes the construction of the circle graph $H$.

We now claim that $G$ has a $k$-colored clique if and only if $H$ has an
independent dominating set of size at most $2k$.

Indeed, let first $K$ be a $k$-colored clique in $G$ containing vertices
$x^{1}_{j_1}, x^{2}_{j_2}, \ldots, x^{k}_{j_k}$, and let us obtain from $K$ an
independent dominating set $S$ in $H$. For $1 \leq i \leq k$, the set $S$
contains the two chords $l^{i}_{j_i}$ and $r^{i}_{j_i}$ from $H^i$. Note that
$S$ is indeed an independent set. For $1 \leq i \leq k$, since both $L^i$ and
$R^i$ are cliques, all the chords in $L^i$ and $R^i$ are dominated by $S$.
Clearly, all parallel chords are also dominated by $S$. The only outer chords
with one endpoint in $H^i$ which are neither dominated by $l^{i}_{j_i}$ nor by
$r^{i}_{j_i}$ are those with its endpoint in the interval $v^{i}_{j_i}$. Let
$c$ be such an outer chord, and suppose that the other endpoint of $c$ is in
$H^{\ell}$. As $K$ is a clique in $G$, it follows that there is no outer chord
in $H$ with one endpoint in $v^{i}_{j_i}$ and the other in
$v^{\ell}_{j_{\ell}}$, and therefore necessarily the chord $c$ is dominated
either by $l^{i}_{j_i}$ or by $r^{\ell}_{j_{\ell}}$.

Conversely, assume that $H$ has an independent dominating set $S$ with $|S|
\leq 2k$. Note that for $1 \leq i \leq k$, because of the two sets of $2k+1$
parallel chords in $H^i$, at least one of the chords in $L^i$ and at least one
of the chords in $R^i$ must belong to $S$, so $|S| \geq 2k$. Therefore, it
follows that $|S|=2k$ and that $S$ contains in $H^i$, for $1 \leq i \leq k$, a
pair of non-crossing chords in $L^i$ and $R^i$. Note that in each $H^i$, the
two chords belonging to $S$ must leave uncovered at least one of the intervals
(corresponding to vertices) $v^{i}_1, \ldots,v^{i}_n$. Let $v^{i}_{j_{i}}$ and
$v^{\ell}_{j_{\ell}}$ be two uncovered vertices in two distinct intervals
$I^{i}$ and $I^{\ell}$, respectively. By the construction of $H$, it holds that
the vertices $x^{i}_{j_{i}}$ and $x^{\ell}_{j_{\ell}}$ must be adjacent in $G$,
as otherwise the outer chord in $H$ between the intervals $v^{i}_{j_{i}}$ and
$v^{\ell}_{j_{\ell}}$ would not be dominated by $S$. Hence, a $k$-colored
clique in $G$ can be obtained by selecting in each $H^i$ any of the uncovered
vertices.\end{proof}

%

The construction of Theorem~\ref{th:independent} can be appropriately modified
to deal with the case when the dominating set is required to induce an acyclic
subgraph.

\begin{theorem}
\label{th:acyclic} \textsc{Acyclic Dominating Set} is $W[1]$-hard in circle
graphs.
\end{theorem}
\begin{proof} As in Theorem~\ref{th:independent}, the reduction is again from
$k$-\textsc{Colored Clique}. From a $k$-colored $G$, we build a circle graph
$H$ that contains all the chords defined in the proof of
Theorem~\ref{th:independent}, plus the following ones for each $H^i$, $1 \leq i
\leq k$ (cf. Fig.~\ref{fig:acyclic} for an illustration): we add another set of
$2k+1$ parallel chords whose left (resp. right) endpoints are placed exactly
before (resp. after) the right (resp. left) endpoint of $l^{i}_1$ (resp.
$r^{i}_n$); cf. the middle green chords in Fig.~\ref{fig:acyclic}. We call
these three sets of $2k+1$ chords \emph{parallel} chords.

\begin{figure}[t]
\flushleft \vspace{-1.2cm}
\includegraphics[width=1.00\textwidth]{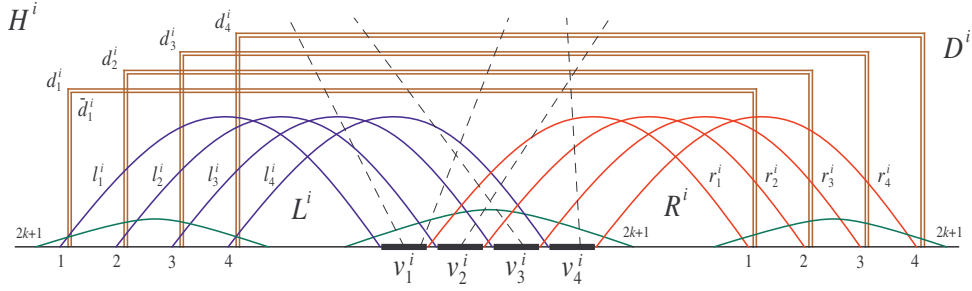}
\vspace{-.7cm} \caption{Gadget $H^i$ in interval $I^i$ used in the proof of
Theorem~\ref{th:acyclic}, corresponding to a color class $X^i$ of the
$k$-colored input graph $G$. The dashed chords correspond to non-edges of
$G$.\label{fig:acyclic}}\vspace{-.3cm}
\end{figure}

Furthermore, we add a clique with $n$ chords $d^{i}_1,\ldots,d^{i}_n$ such that
for $1 \leq j \leq n$ the left (resp. right) endpoint of $d^{i}_j$ is placed
exactly after the left (resp. right) endpoint of $l^{i}_j$ (resp. $r^{i}_j$).
Finally, for each such a chord $d^{i}_j$ we add a parallel twin chord, denoted
by $\bar{d}^{i}_j$. We call these $2t$ chords \emph{distance} chords, and their
union is denoted by $D^i$; cf. the brown edges in Fig.~\ref{fig:acyclic}. This
completes the construction of $H$. Note that a pair of chords $l^{i}_{j_1}$ and
$r^{i}_{j_2}$ dominates all the distance chords in $H^i$ if and only if
$l^{i}_{j_1}$ and $r^{i}_{j_2}$ do not cross, that is, if and only if $j_1 \leq
j_2$.

We now claim that $G$ has a $k$-colored clique if and only if $H$ has an
acyclic dominating set of size at most $2k$.

Indeed, let first $K$ be a $k$-colored clique in $G$. An independent (hence,
acyclic) dominating set $S$ in $H$ of size $2k$ can be obtained from $K$
exactly as explained in the proof of Theorem~\ref{th:independent}. Note that in
each $H^i$, the distance chords in $D^i$ are indeed dominated by $S$ because
the corresponding chords in $L^i$ and $R^i$ do not cross.

Conversely, assume that $H$ has an acyclic dominating set $S$ with $|S| \leq
2k$. First assume that $S$ contains no outer chord. By the parallel chords in
each $H^i$ (cf. the green chords in Fig.~\ref{fig:acyclic}), it is easy to
check that $S$ must contain at least two chords in each $H^i$, and therefore we
have that $|S| = 2k$. We now distinguish several cases according to which two
chords in a generic $H^i$ can belong to $S$. Let $\{u,v\} = S \cap V(H^i)$.
Because of the parallel chords, it is clear that only one $L^i$, $R^i$, or
$D^i$ cannot contain both $u$ and $v$. It is also clear that no parallel chord
can be in $S$. If $u \in D^i$ and $v \in L^i$, let w.l.o.g. $u = d^{i}_{j_1}$
and $v = l^{i}_{j_2}$. Since $v$ must dominate the twin chord
$\bar{d}^{i}_{j_1}$, it follows by the construction of $H^i$ that $j_2 \leq
j_1$, and therefore the chord $r^{i}_{j_1}$ is dominated neither by $u$ nor by
$v$ (cf. Fig.~\ref{fig:acyclic}), a contradiction. The case $u \in D^i$ and $v
\in R^i$ is similar. Therefore, we may assume w.l.o.g. that $u=l^{i}_{j_1}$ and
$u=r^{i}_{j_2}$. Note that such a pair of chords $l^{i}_{j_1}$ and
$r^{i}_{j_2}$ dominates all the distance chords in $H^i$ if and only if
$l^{i}_{j_1}$ and $r^{i}_{j_2}$ do not cross.
Hence, as in the proof of Theorem~\ref{th:independent}, for each $H^i$, $1 \leq
i \leq k$, the two chords belonging to $S$ leave at least one uncovered
interval $v^{i}_{j_i}$ (corresponding to vertex $x^{i}_{j_i}$), and in order
for all the outer chords in $H$ to be dominated, the union of the $k$ uncovered
vertices must induce a $k$-colored clique in $G$. Therefore, if $H$ has an
acyclic dominating set of size at most $2k$ with no outer chord, then $G$ has a
$k$-colored clique. Note that in this case $S$ consists of an independent set.

Otherwise, the acyclic dominating set $S$ contains some outer chord. Assume
w.l.o.g. that $S$ contains outer chords with at least one endpoint in each of
$H^1,\ldots,H^p$ (with $p \geq 2$, as each chord has two endpoints), and no
outer chord with an endpoint in any of $H^{p+1},\ldots,H^{k}$ (only if $p<k$).
By the arguments above, for $p+1 \leq i \leq k$ it follows that $S$ contains
exactly one chord in $L^i$ and exactly one chord in $R^i$. For $1 \leq i \leq
p$, in order for all the parallel chords in $H^i$ to be covered, $S$ must
contain some chords in $L^i$, $R^i$, or $D^i$. As by assumption $|S| \leq 2k$,
in at least one $H^i$ with $1 \leq i \leq p$, $S$ contains exactly one chord in
$L^i$, $R^i$, or $D^i$. By the construction of $H$, this chord must necessarily
be a distance chord, as otherwise some parallel chords in $H^i$ would not be
dominated by $S$ (cf. Fig.~\ref{fig:acyclic}). Assume w.l.o.g. that in $H^1$
only the distance chord $d^{1}_{j_1}$ and one outer chord outgoing from
interval $v^{1}_{j_2}$ belong to $S$. But then if $j_2 \geq j_1$ (resp. $j_2 <
j_1$) the chord $r^{1}_{j_2}$ (resp. $l^{1}_{j_2}$) is not dominated by $S$, a
contradiction. We conclude that $S$ contains at least two outer chords in
$H^1$.

By a simple counting argument, as $|S| \leq 2k$ it follows that for $1 \leq i
\leq p$, $S$ contains exactly one distance chord and two outer chords from each
of $H^1,\ldots,H^p$ (and, in particular, $|S| = 2k$). But then the subgraph of
$H[S]$ induced by the chords belonging to $V(H^1)\cup \ldots \cup V(H^p)$ has
minimum degree at least two, and therefore it contains a cycle, a contradiction
to the assumption that $H[S]$ is acyclic. Thus, $S$ cannot contain any outer
chord, and the theorem follows.\end{proof}

%

%

\subsection{NP-completeness for a given tree}
\label{sec:hardTrees}

The last result of this section is the \textsc{NP}-completeness when the
dominating set is restricted to be isomorphic to a {\sl given} tree.

\begin{theorem}
\label{thm:treeHARD} Let $T$ be a given tree. Then \textsc{$\{T\}$-Dominating
Set} is \textsc{NP}-complete in circle graphs when $T$ is part of the input.
\end{theorem}
\begin{proof}
We present a reduction from the 3-\textsc{Partition} problem, which consists in
deciding whether a given multiset of $n=3m$ integers $I$ can be partitioned
into $m$ triples that all have the same sum $B$. The 3-\textsc{Partition}
problem is strongly \textsc{NP}-complete, and in addition, it remains
\textsc{NP}-complete even when every integer in $I$ is strictly between $B/4$
and $B/2$~\cite{GaJo79}. Let $I=\{a_1,\ldots,a_n\}$ be an instance of
3-\textsc{Partition}, in which we can assume that the $a_i$'s are between $B/4$
and $B/2$, and let $B = \sum_{i=1}^{n}a_i/m$ be the desired sum. Note that we
can also assume that $B$ is an integer, as otherwise $I$ is obviously a
\textsc{No}-instance.

We proceed to define a tree $T$ and to build a circle graph $G$ that has a
$\{T\}$-dominating set $S$ if and only if $I$ is a \textsc{Yes}-instance of
3-\textsc{Partition}. Given $I=\{a_1,\ldots,a_n\}$, let $T$ be the rooted tree
obtained from a root $r$ to which we attach a path with $a_i$ vertices, for
$i=1,\ldots,n$; see Fig~\ref{fig:treeHARD}(a) for an example with $n=9$, $m=3$,
and $B=5$. (In this figure, for simplicity not all the $a_i$'s are between
$B/4$ and $B/2$, but we assume that this fact is true in the proof.) Note that
$|V(T)|=mB+1$.

The circle graph $G$ is obtained as follows; see Fig~\ref{fig:treeHARD}(b) for
the construction corresponding to the instance of Fig~\ref{fig:treeHARD}(a): We
start with a chord $r$ that will correspond to the root of $T$. Now we add $mB$
parallel chords $g_1,\ldots,g_{mB}$ intersecting only with $r$. These chords
are called \emph{branch} chords; cf. the green chords in
Fig~\ref{fig:treeHARD}(b). We can assume that the endpoints of the branch
chords are ordered clockwise in the circle. For $i=1,\ldots,mB$, we add a chord
$b_i$ incident only with $g_i$. These chords are called \emph{pendant} chords;
cf. the blue chords in Fig~\ref{fig:treeHARD}(b), where for better visibility
these chords have been depicted outside the circle. Finally, for $i \in
\{1,2,\ldots,mB\} \setminus \{B,2B,\ldots,mB\}$, we add a chord $r_i$ whose
first endpoint is exactly after the first endpoint of $b_i$ (in the
anticlockwise order, starting from any of the endpoints of the root $r$), and
whose second endpoint is exactly before the second endpoint of $b_{i+1}$. These
chords are called \emph{chain} chords; cf. the red chords in
Fig~\ref{fig:treeHARD}(b). Note that $r_i$ is adjacent to $g_i,g_{i+1},b_{i}$,
and $b_{i+1}$. This completes the construction of the circle graph $G$. Each
one of the $m$ connected components that remain in $G$ after the removal of $r$
and the
parallel chords is called a \emph{block}. 

\begin{figure}[bt]
\flushleft \vspace{-.3cm}
\includegraphics[width=1.10\textwidth]{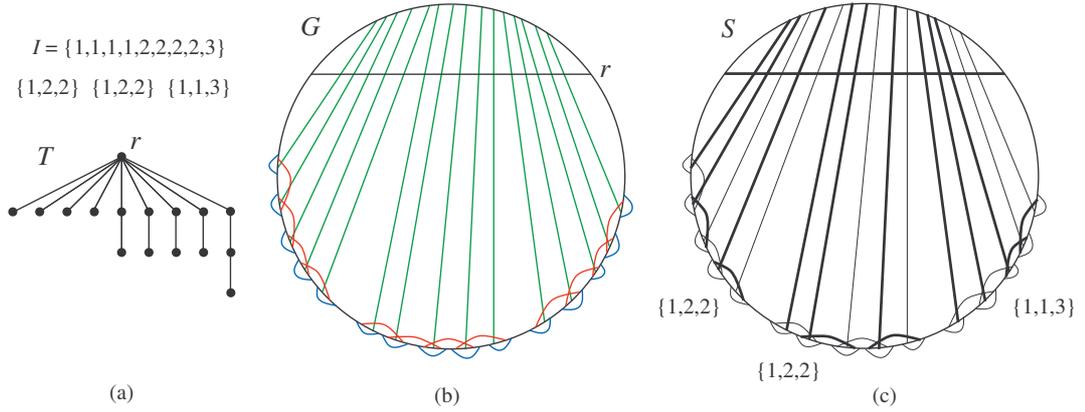}
\vspace{-.2cm} \caption{Reduction in the proof of Theorem~\ref{thm:treeHARD}.
(a) Instance $I$ of 3-\textsc{Partition}, with $n=9$, $m=3$, and $B=5$,
together with the associated tree $T$. (b). Circle graph $G$ built from $I$.
(c) The thick chords define a $\{T\}$-dominating set $S$ in
$G$.\label{fig:treeHARD}}\vspace{-.5cm}
\end{figure}

Let first $I$ be a \textsc{Yes}-instance of 3-\textsc{Partition}, and we
proceed to define a $\{T\}$-dominating set $S$ in $G$. For $1 \leq j \leq m$,
let $B_j=\{a_1^j,a_2^j,a_{3}^j\}$ be the $j$-th triple of the 3-partition of
$I$; in the instance of Fig~\ref{fig:treeHARD}(a), we have $B_1=\{1,2,2\},
B_2=\{1,2,2\}, B_3=\{1,1,3\}$. We include the chord $r$ in $S$, plus the
following chords for each $j \in \{1,,\ldots,m\}$: For $i \in \{1,2,3\}$, we
add to $S$ the branch chord $g_{(j-1)B + \sum_{k=1}^{i-1}a_k^j+1}$ plus, if
$a_i^j \geq 2$, the chain chords $r_{(j-1)B + \sum_{k=1}^{i-1}a_k^j+\ell}$ for
$\ell \in \{1,\ldots,a_{i}^{j}-1\}$; cf. the thick chords in
Fig~\ref{fig:treeHARD}(c). It can be easily checked that $S$ is a
$\{T\}$-dominating set of $G$.

Conversely, let $S$ be a $\{T\}$-dominating set $S$ in $G$, and note that we
can assume that the root of $T$ has arbitrarily big degree. As the vertex of
$G$ corresponding to the chord $r$ is the only vertex of $G$ of degree more
than 6,
necessarily $r$ belongs to $S$, and corresponds to the root of $T$. 

We claim that $S$ contains no pendant chord. Indeed, by construction of $G$,
exactly $n$ of the branch chords are in $S$, which dominate exactly $n$ pendant
chords. As $G[S \setminus \{r\}]$ consists of $n$ disjoint paths, each attached
to $r$ through a branch chord, the total number of chords in these paths which
are not branch chords is $mB-n$. These $mB-n$ pendant or chain chords must
dominate the pendant chords that are not dominated by branch chords, which are
also $mB -n$ many. Assume that a pendant chord $b$ belongs to $S$. Since $T$ is
a tree, there must exist a path $P$ in $S$ between $b$ and one of the branch
chords, say $g$. Assume that $P$ contains $p$ chords, including $b$ but not
$g$. It is clear that $b$ is the only pendant chord contained in $P$, as
otherwise $P$ would have a cycle. Therefore, $P$ has $p$ chords and dominates
exactly $p-1$ pendant chords that are not dominated by branch chords, which
contradicts the fact that $mB-n$ pendant or chain chords must dominate the
$mB-n$ pendant chords that are not dominated by branch chords. Hence, $S$
contains no pendant chord, so $S$ contains exactly $mB-n$ chain chords.

Since $T$ is a tree, each path in $S$ made of consecutive chain chords
intersects exactly one branch chord. As the $a_i$'s are strictly between $B/4$
and $B/2$, each block has exactly 3 branch chords in $S$. The fact that chain
chords are missing between consecutive blocks assures the existence of a
3-partition of $I$. More precisely, the restriction of $S$ to each block
defines the integers belonging to each triple of the 3-partition of $I$ as
follows. For a branch chord $g_i \in S$, let $P_i$ be the path in $S$ hanging
from $g_i$, which consists only of chain chords. Then, for each branch chord
$g_i \in S$, the corresponding integer is defined by the number of vertices in
$P_i$ plus one. By the above discussion, these $m$ triples define a 3-partition
of $I$. The theorem follows.
\end{proof}

To conclude this section, it is worth noting here that
\textsc{$\{T\}$-Dominating Set} is $W[2]$-hard in general graphs. This can be
proved by an easy reduction from \textsc{Set Cover} parameterized by the number
of sets, which is $W[2]$-hard~\cite{PaMo81}. Indeed, let $\mathcal{C}$ be a
collection of subsets of a set $S$, and the question is whether there exist at
most $k$ subsets in $\mathcal{C}$ whose union contains all elements of $S$. We
construct a graph $G$ as follows. First, we build a bipartite graph $(A \cup B,
E)$, where there is a vertex in $A$ (resp. $B$) for each subset in
$\mathcal{C}$ (resp. element in $S$), and there is an edge in $E$ between a
vertex in $A$ and a vertex in $B$ if the corresponding subset contains the
corresponding element. We add a new vertex $v$, which we join to all the
vertices in $A$, and $k+1$ new vertices joined only to $v$. It is then clear
that $G$ has a dominating set isomorphic to a star with exactly $k$ leaves if
and only if there is a collection of at most $k$ subsets in $\mathcal{C}$ whose
union contains all elements of $S$.


%
%



\section{Polynomial and FPT algorithms}
\label{sec:algo}

In this section we provide polynomial and FPT algorithms for finding dominating
sets in a circle graph which are isomorphic to trees. Namely, in
Section~\ref{sec:algoTrees} we give a polynomial-time algorithm to find a
dominating set isomorphic to {\sl some} tree. This algorithm contains the main
ideas from which the other algorithms in this section are inspired. In
Section~\ref{sec:FPTalgoTrees} we modify the algorithm to find a dominating set
isomorphic to a {\sl given} tree $T$ in \textsc{FPT} time, the parameter being
the size of $T$. By carefully analyzing its running time, we prove that this
\textsc{FPT} algorithm runs in {\sl subexponential} time. It follows from this
analysis that if the given tree $T$ has bounded degree (in particular, if it is
a path), then the problem of find a dominating set isomorphic to $T$ can be
solved in polynomial time.





\subsection{Polynomial algorithm for trees}
\label{sec:algoTrees}

Note that, in contrast with Theorem~\ref{treepoly} below,
Theorem~\ref{th:acyclic} in Section~\ref{sec:hardness2} states that, if
$\mathcal{F}$ is the set of all forests, then \textsc{$\mathcal{F}$-Dominating
Set} is $W[1]$-hard in circle graphs. 
This is one of the interesting examples where the fact of imposing connectivity
constraints in a given problem makes it computationally easier, while it is
usually not the case (see for instance~\cite{APP+08,RST10}).

\begin{theorem}\label{treepoly}
Let $\mathcal{T}$ be the set of all trees. Then
$\mathcal{T}$-\textsc{Dominating Set} can be solved in polynomial time in
circle graphs. In other words, \textsc{Connected Acyclic Dominating Set} can be
solved in polynomial time in circle graphs.
\end{theorem}

%
%
%

\begin{proof}
Let $C$ be a circle graph on $n$ vertices and let $\mathcal{C}$ be an arbitrary
circle
representation of $C$. 
We denote by $\mathcal{P}$ the set of intersections of the circle and the
chords in this representation. The elements of $\mathcal{P}$ are called
\emph{points}. Without loss of generality, we can assume that only one chord
intersects a given point. Given two points $a,b \in \mathcal{P}$, the
\emph{interval $[a,b]$} is the interval from $a$ to $b$ in the anticlockwise
order. Given four (non-necessarily distinct) points $a,b,c,d \in \mathcal{P}$,
with $a \leq c \leq d \leq b$, by the \emph{region $ab-cd$} we mean the union
of the two intervals $[a,c]$ and $[d,b]$. Note that these two intervals can be
obtained by ``subtracting'' the interval $[c,d]$ from the interval $[a,b]$;
this is why we use the notation $ab-cd$.

In the following, by \emph{size} of a set of chords, we mean the number of
chords in it, i.e., the number of vertices of $C$ in this set. We say that a
forest $F$ of $C$ \emph{spans} a region $ab-cd$ if each of $a,b,c$, and $d$ is
an endpoint of some chord in $F$, and each endpoint of a chord of $F$ is either
in $[a,c]$ or in $[d,b]$. A forest $F$ is \emph{split} by a region $ab-cd$  if
for each connected component of $F$ there is exactly one chord with one
endpoint in $[a,c]$ and one endpoint in $[d,b]$. Given a region $ab-cd$, a
forest $F$ is $(ab-cd)$-\emph{dominating} if all the chords of $C$ with both
endpoints either in the interval $[a,c]$ or in the interval $[d,b]$ are
dominated by $F$. A forest is \emph{valid} for a region $ab-cd$ if it spans
$ab-cd$, is split by $ab-cd$, and is $(ab-cd)$-dominating. 




Note that an $(ab-cd)$-dominating forest with several connected components
might not dominate some chord going from $[a,c]$ to $[d,b]$. This is not the
case if $F$ is connected, as stated in the following claim.

\begin{claimN}\label{extraforTree}
Let $T$ be a valid tree for a region $ab-cd$. Then all the chords of $C$ with
both endpoints in $[a,c] \cup [d,b]$ are dominated by $T$.
\end{claimN}
\begin{proof}
All the chords with both endpoints either in $[a,c]$ or in $[d,b]$ are
dominated by $T$, since $T$ is $(ab-cd)$-dominating. Hence we just have to
prove that the chords with one endpoint in $[a,c]$ and one in $[d,b]$ are
dominated by $T$. Since $T$ spans $ab-cd$, there are in $T$ a chord $\gamma$
with endpoint $a$ and a chord $\gamma'$ with endpoint $c$. Since $T$ is split
by $ab-cd$, there is a unique chord $uv$ in $T$ with one endpoint in $[a,c]$
and one in $[d,b]$. In $T \backslash \{uv \}$, $\gamma$ and $\gamma'$ are in
the same connected component. Indeed, otherwise their connected components span
two disjoint intervals of $[a,c]$. But $uv$ is the unique chord of $T$ with one
endpoint in $[a,c]$ and one in $[d,b]$, thus $uv$ cannot connect these
components. So, if $T$ is a tree, $\gamma$ and $\gamma'$ are in the same
connected component.

Thus for each point $p$ in $[a,c]$, there is a chord $ef$ of the connected
component of $\gamma$ and $\gamma'$ such that $a \leq e \leq p \leq f \leq c$.
Therefore, the chords with one endpoint in $[a,c]$ and one endpoint in $[d,b]$
are dominated.
\end{proof}

We now state two properties that will be useful in the algorithm. Their
correctness is proved below.

\begin{itemize}
\item[\textbf{T1}]\label{r1}
Let $F_1$ and $F_2$ be two valid forests for two regions $ab-cd$ and $ef-gh$,
respectively, such that $a \leq c \leq e \leq g \leq h \leq f \leq d \leq b$.
If there is no chord with both endpoints either in $[c,e]$ or in $[f,d]$, then
$F_1 \cup F_2$ is valid for $ab-gh$ (see Fig. \ref{treedom}).


\item[\textbf{T2}]\label{r2} Let $F_1$ and $F_2$ be two
valid forests for two regions $ab-cd$ and $ef-gh$, respectively ($F_2$ being
possibly empty), and let $uv$ be a chord such that $u \leq a \leq c \leq e \leq
g \leq v \leq h \leq f \leq d \leq b$, and such that there is no chord with
both endpoints either in $[u,a]$, or in $[g,v]$, or in $[v,h]$, or in $[b,u]$.
Then $F_1 \cup F_2 \cup \{uv\}$ is a tree which is valid for $df-ce$. When
$F_2$ is empty, we consider that $e,f,g,h$ correspond to the point $v$. (see
Fig. \ref{treedom}).
\end{itemize}

\begin{figure}[t]
\center \vspace{-.4cm}
\input{treedom.pstex_t} \caption{On the left (resp. right),
regions corresponding to Property~\textbf{T1} (resp. Property~\textbf{T2}).
Full lines correspond to real chords of $C$, dashed lines correspond to the
limit of regions. Bold intervals correspond to intervals with no chord of $C$
with both endpoints in the interval.} \label{treedom}
\end{figure}
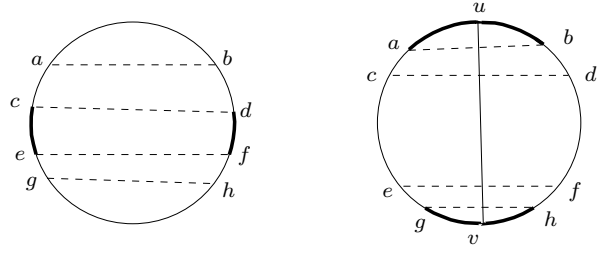

Roughly speaking, the intuitive idea behind this two properties is to reduce
the length of the circle in which we still have to do some computation (that
is, outside the valid regions), which will be helpful in the dynamic
programming routine. Again, the proof is structured along a series of claims.
Before verifying the correctness of Properties~\textbf{T1} and~\textbf{T2}, let
us first state a useful general fact.


\begin{claimN}\label{oneeachside}
Let $ab-cd$ be a region and let $F$ be a valid forest for $ab-cd$. The chords
with one endpoint in $[c,d]$ and one endpoint in $[d,c]$ are dominated by $F$.
\end{claimN}
\begin{proof}
Let us now consider a chord $\gamma$ with one endpoint $\alpha$ in $[c,d]$ and
one endpoint $\beta$ in $[d,c]$. First assume that $\beta$ is in $[b,a]$. Since
$F$ is split by $ab-cd$, there is a chord of $F$ with one endpoint in $[a,c]$
and one in $[d,b]$, and such a chord dominates $\gamma$. Therefore, by
symmetry, we can assume that $\beta$ is in $[a,c]$. Since $F$ spans $ab-cd$,
there is a chord in $F$ with endpoint $c$. Since $F$ is split by $ab-cd$, there
is a chord $\omega=uv$ of $F$, in the same connected component as the chord
with endpoint $c$, with one endpoint in $[a,c]$ and one endpoint in $[d,b]$. If
$\beta \leq u \leq \alpha \leq v$, then the chord $\gamma$ is dominated by $F$.
Thus we can assume that $u \leq \beta < \alpha$. And note that by assumption,
$\beta \leq c \leq \alpha$. But in $F$, the chord with endpoint $c$ is
connected to the chord $\omega$, thus there is a chord $wz$ of $F$ such that $w
\leq \beta \leq z \leq \alpha$, and therefore the chord $\gamma$ is dominated
by $F$, which achieves the proof of the claim.\end{proof}

Note that the proof of Claim~\ref{oneeachside} is symmetric, and then the same
result is still true for the intervals $[a,b]$ and $[b,a]$. Note also that the
same result holds without the asumption that $F$ is $(ab-cd)$-dominating.

\begin{claimN}
\label{cl:T1correct} Property~\textbf{\emph{T1}} is correct.
\end{claimN}
\begin{proof}
Let us prove that $F_1 \cup F_2$ is a forest and that it is valid for the
region $ab-gh$. For an illustration refer to Fig.~\ref{treedom}. Since $F_1$
and $F_2$ span $ab-cd$ and $ef-gh$ respectively, all the endpoints of the
chords of $F_1$ are in $[a,c] \cup [d,b]$, and those of $F_2$ are in $[e,g]
\cup [h,f]$. Thus the order $a \leq c \leq e \leq g \leq h \leq f \leq d \leq
b$ ensures that a chord of $F_1$ cannot cross of chord of $F_2$. Therefore,
$F_1 \cup F_2$ is still a forest and the connected components of the union are
precisely the connected components of $F_1$ and the connected components of
$F_2$.

Since $F_1$ spans $ab-cd$, there is a chord of $F_1$ with endpoint $a$ and one
chord with endpoint $b$, and the same holds for $F_2$ and $g,h$. Then $F_1 \cup
F_2$ spans the region $ab-gh$.

Since $F_1$ is split by $ab-cd$, there is exactly one chord per connected
component between $[a,c]$ and $[d,b]$, thus also between $[a,g]$ and $[h,b]$.
The same holds for $F_2$. Thus each connected component of $F_1 \cup F_2$ has
exactly one chord with one endpoint in $[a,g]$ and the other one in $[h,b]$.
So $F_1 \cup F_2$ is split by $ab-gh$.

Let us now prove that $F_1 \cup F_2$ is $(ab-gh)$-dominating. Let us verify
that all the chords in the interval $[a,g]$ are dominated by $F_1 \cup F_2$. By
symmetry, the same will hold for the interval $[h,b]$. All the chords with both
endpoints in the interval $[a,c]$ are dominated by $F_1$, and those with both
endpoints in the interval $[e,g]$ are dominated by $F_2$. By assumption, there
is no chord in the interval $[c,e]$. The chords with one endpoint in $[c,d]$
and one endpoint in $[d,c]$ are dominated by $F_1$ by Claim~\ref{oneeachside},
and those with one endpoint in $[e,f]$ and one endpoint in $[f,e]$ are
dominated by $F_2$. Thus all the chords with both endpoints in $[a,c]$ are
dominated by $F_1 \cup F_2$, which ensures that $F_1 \cup F_2$ is
$(ab-gh)$-dominating.

Therefore,  $F_1 \cup F_2$ is valid for the region $ab-gh$.
\end{proof}

\begin{claimN}
\label{cl:T2correct} Property~\textbf{\emph{T2}} is correct.
\end{claimN}
\begin{proof}
Let $F_1$  and $F_2$ be two valid forests for $ab-cd$ and for $ef-gh$,
respectively ($F_2$ being possibly empty), and let $uv$ be a chord with
endpoints $u$ and $v$, such that $u \leq a \leq c \leq e \leq g \leq v \leq h
\leq f \leq d \leq b$ and such that there is no chord with both endpoints in
either $[u,a]$, or $[g,v]$, or $[v,h]$, or $[b,u]$. For an illustration, refer
also to Fig.~\ref{treedom}.

First note that $T=F_1 \cup F_2 \cup \{uv\}$ is a tree. Indeed, as in
Claim~\ref{cl:T1correct}, one can prove that $F_1 \cup F_2$ is a forest with
exactly one chord with one endpoint in $[a,g]$ and one endpoint in $[h,b]$ per
connected component. Thus, the addition of $uv$ ensures that $T$ is a tree.

Since $F_1$ and $F_2$ spans $ab-cd$ and $ef-gh$ respectively, $T$ spans
$df-ce$. Indeed, there are chords intersecting $d$ and $c$ in $F_1$, chords
intersecting $e$ and $f$ in $F_2$, and all the chords are strictly inside
$[d,c] \cup [e,f]$. Note that when the forest $F_2$ is empty, there is a chord
intersecting $v$, and thus the tree $T$ spans $df-ce$.

The tree $T$ spans $df-ce$, since there is exactly one chord with one endpoint
in $[d,c]$ and one endpoint in $[e,f]$, which is precisely the chord $uv$.

Let us prove that $T$ is $(df-ce)$-dominating.
First note that the chords with one endpoint in $[u,v]$ and one endpoint in
$[v,u]$ are dominated by $uv$. By symmetry, we just have to prove that the
chords with both endpoints in $[u,v]$ are dominated by $T$. By symmetry again,
we just have to prove that the chords with both endpoints in $[u,c]$ are
dominated. There is no chord with both endpoints in $[u,a]$, the chords with
both endpoints in $[a,c]$ are dominated by $F_1$, since $F_1$ is
$(ab-cd)$-dominating. By Claim~\ref{oneeachside}, the chords with one endpoint
in $[a,b]$ and one endpoint in $[b,a]$ are dominated by $F_1$, thus the chords
with one endpoint in $[u,a]$ and one in $[a,c]$ are dominated.

Therefore, we conclude that $T$ is a valid tree for $df-ce$.
\end{proof}

For a region $ab-cd$, we denote by $v_{ab,cd}^f$ (resp. $v_{ab,cd}^t$) the
least integer $l$ for which there is a valid forest (resp. tree) of size $l$
for $ab-cd$. If there is no valid forest (resp. tree) for $ab-cd$, we set
$v_{ab,cd}^f = + \infty$ (resp. $v_{ab,cd}^t = + \infty$). Let us now describe
our algorithm based on dynamic programming. With each region $ab-cd$, we
associate two integers $v_{ab,cd}^1$ and $v_{ab,cd}^2$. Algorithm~\ref{algTree}
below calculates these two values for each region. We next prove that
$v_{ab,cd}^1=v_{ab,cd}^f$ and $v_{ab,cd}^2=v_{ab,cd}^t$, and that
Algorithm~\ref{algTree} computes the result in polynomial time.

\begin{algorithm}[htb]
\caption{Dynamic programming for computing a dominating tree} \label{algTree}
\begin{algorithmic} 

\medskip

\STATE{\textbf{for} each region $ab-cd$ $\ $\textbf{do}$\ $ $v_{ab,cd}^1
\leftarrow \infty$; \ $v_{ab,cd}^2 \leftarrow \infty$}

\STATE{\textbf{for} each chord $ab$ of the circle graph \ \textbf{do} \
$v_{ab,ab}^1 \leftarrow 1$; \ $v_{ab,ab}^2 \leftarrow 1$}

\FOR {$j=2$ to $n$}

\IF{there are two regions $ab-cd$ and $ef-gh$ such that $v_{ab,cd}^1=j_1$ and
$v_{ef,gh}^1=j_2$ with $j_1+j_2=j$ satisfying Property~\textbf{T1}, with
$v_{ab,gh}^1=+ \infty$}

\STATE{$v_{ab,gh}^1 \leftarrow j$}
     \ENDIF

\IF{there is a region $ab-cd$ and a chord $uv$ such that $v_{ab,cd}^1=j-1$
satisfying Property~\textbf{T2} with an empty second forest}

\IF {$v_{dv,cv}^1=+ \infty$} \STATE{$v_{dv,cv}^1 \leftarrow j$}
     \ENDIF
\IF{$v_{dv,cv}^2=+ \infty$} \STATE{$v_{dv,cv}^2 \leftarrow j$} \ENDIF
     \ENDIF

\IF{there are two regions $ab-cd$ and $ef-gh$ and a chord $uv$ such that
$v_{ab,cd}^1=j_1$ and $v_{ef,gh}^1=j_2$ with $j_1+j_2=j-1$ satisfying
Property~\textbf{T2}}

\IF {$v_{df,ce}^1=+ \infty$} \STATE{$v_{df,ce}^1 \leftarrow j$}
     \ENDIF
\IF{$v_{df,ce}^2=+ \infty$} \STATE{$v_{df,ce}^2 \leftarrow j$} \ENDIF
     \ENDIF\ENDFOR
\end{algorithmic}
\end{algorithm}

\begin{claimN}
\label{cl:inizialitation} For any region $ab-cd$, $v_{ab,cd}^1= 1$ (resp.
$v_{ab,cd}^2= 1$) if and only if $v_{ab,cd}^f = 1$ (resp. $v_{ab,cd}^t = 1$).
\end{claimN}
\begin{proof}
Let $ab-cd$ be a region such that $v^f_{ab,cd}=1$. Therefore, there is a set of
chords of size one which is valid for $ab-cd$. Let $\omega$ be this chord.
Since $\omega$ spans $a,b,c$ and $b$, $\omega$ has endpoints $a,b,c,d$. This
implies that $a=c$ and $b=d$, i.e., $\omega$ is precisely the chord $ab$. If
$v^f_{ab,cd}=1$, then $a=c$, $b=d$, and the chord $ab$ exists, which
corresponds exactly to the initialization of the algorithm. Conversely, it is
clear that by definition the chord $ab$ is valid for the region $ab-ab$. Thus,
$v_{ab,cd}^1= 1$ if and only if $v_{ab,cd}^f= 1$.
The same result holds for trees, since a forest of size one is a tree.
\end{proof}


\begin{claimN}
\label{cl:oneSenseOfInequality} For any region $ab-cd$, $v_{ab,cd}^f \leq
v_{ab,cd}^1 $ (resp. $v_{ab,cd}^t \leq v_{ab,cd}^2$).
\end{claimN}
\begin{proof}
The claim is true for the initialization, since if $v_{ab,cd}^1= 1$ then
$v_{ab,cd}^f =1$. By induction it is still true for all integers $k$, since
Properties~\textbf{T1} and~\textbf{T2} are correct, and when a value is
affected in the dynamic programming of Algorithm~\ref{algTree}, one of the two
properties is applied.
\end{proof}

\begin{claimN}
\label{cl:ineqTrees} For any region $ab-cd$, $v_{ab,cd}^t \geq v_{ab,cd}^2$ and
$v_{ab,cd}^f \geq v_{ab,cd}^1$.
\end{claimN}
\begin{proof}
Let us prove it by induction on $j$. Claim~\ref{cl:inizialitation} ensures that
the result is true for $j=1$. Assume that for all $j<k$, if $v_{ab,cd}^f=j$
then $v_{ac,bd}^1 \leq j$ and that if $v_{ac,bd}^t = j$ then $v_{ac,bd}^2 \leq
j$.

Let us first prove that the induction step holds for trees. We now prove that
if $v_{ab,cd}^t=j$, then $v_{ab,cd}^2 \leq j$. Let $T$ be a valid tree of size
$j$ for the region $ab-cd$. Since $T$ spans $ab-cd$, there is exactly one chord
$uv$ with one endpoint in $[a,c]$ and one endpoint in $[d,b]$. Let $F_1$ be the
restriction of $T$ to the chords with both endpoints in $ [a,c]$, and let $F_2$
be the restriction of $T$ to the chords with both endpoints in $[d,b]$. Note
that $T= F_1 \cup F_2 \cup \{uv\}$. Let $e,f$ (resp. $g,h$) be the points of
the circle graph intersected by $F_1$ (resp. $F_2$) such that $a \leq e \leq u
\leq f \leq c$ (resp. $d \leq h \leq v \leq g \leq b$), and $e,f$ (resp. $g,h$)
are as near as possible from $u$ (resp. $v$) (see Fig. \ref{progdyna1} for an
example). Let us denote by $j_1$ (resp. $j_2$) the size of $F_1$ (resp. $F_2$).
Note that $j_1+j_2=j-1$.

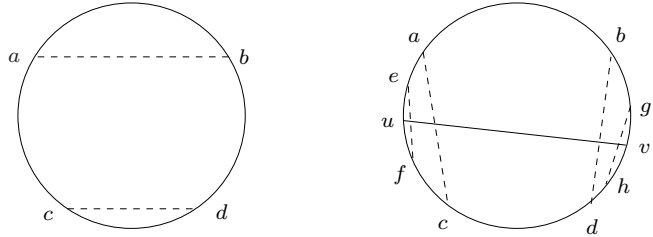
\begin{figure}
\center \input{progdyna1.pstex_t} \caption{On the left the original tree, and
on the right the partition of the tree into two forests and a chord $uv$.}
\label{progdyna1}
\end{figure}

Let us prove that $F_1$ is valid for $ac-ef$, that $F_2$ is valid for $db-hg$,
and that Property~\textbf{T2} can be applied to $F_1$, $F_2$, and the chord
$uv$. By symmetry, we just have to prove that $F_1$ is valid. There are chords
intersecting $e$, $f$ by definition of $e,f$, and chords intersecting $a,c$
since $T$ spans $ab-cd$. The endpoints of the chords are in $[a,e] \cup [f,c]$,
since $T$ spans $ab-cd$ and $e,f$ are the nearest points from $u$ which are in
$T$. Thus $F_1$ spans $ac-ef$.

If there is a connected component of $F_1$ with no chord from $[a,e]$ to
$[f,c]$, then $F_1 \cup F_2 \cup \{uv\}$ cannot be a tree, since it would not
be connected. If a connected component of $F_1$ has two chords from $[a,e]$ to
$[f,c]$, then $F_1 \cup \{u,v\}$ has a cycle, which contradicts the fact that
$T$ is a tree. Thus $F_1$ spans $ac-ef$.

Let us prove that $F_1$ is $(ac-ef)$-dominating. Indeed, if there is a chord
with both endpoints in $[a,e]$ which is not dominated by $F_1$, it cannot be
dominated by $F_2$ and $uv$, since none of their endpoints is in this interval.
Thus $T$ is not valid. Hence all the chords in the interval $[a,e]$, and by
symmetry also in the interval $[f,c]$, are dominated by $F_1$. Therefore $F_1$
is valid. By induction hypothesis, since the size of $F_1$ is at most $j_1$, we
have $v_{ac,ef}^2 \leq j_1$, and the same holds for $F_2$.

Since $T$ is valid, one can note that there is no chord with both endpoint
either in $[e,u]$, or $[u,f]$, or $[h,v]$, or $[v,g]$. Thus
Property~\textbf{T2} can be safely applied and then $v_{ab,cd}^2 \leq j$, as we
wanted to prove.

\begin{figure}
\center \vspace{-.35cm}
\input{progdyna2.pstex_t} \caption{On the left the
original forest, and on the right the partition of the forest into the two
forests.} \label{progdyna2}
\end{figure}
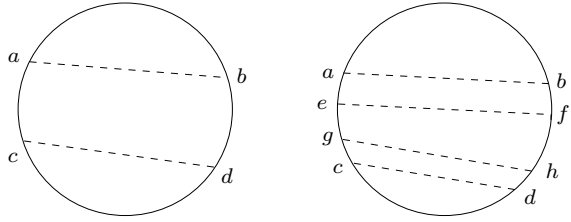

Let us now prove that the induction step also holds for forests. Let us now
consider a forest $F$ for the region $ab-cd$. If the forest has exactly one
connected component, that is, if it is a tree, then the inequality holds by the
first part of the induction.

Assume now that $F$ has at least two connected components. Let us now prove
that if $v_{ab,cd}^f=j$ then $v_{ab,cd}^1 \leq j$. Let $ab-cd$ be a region with
$v_{ac,bd}^f=k$. This means that there exists a forest $F$ with one endpoint in
$a,b,c,d$, since $F$ spans $ab-cd$. Since $v_{ab,cd}^t \neq j$, $F$ has at
least two connected components. Note that the case when $v_{ab,cd}^t = j$ is
treated just above.

Since $F$ spans $ab-cd$, all the endpoints of $F$ are in $[a,c] \cup [d,c]$.
Let $F_1$ be the connected component of the chord with endpoint $a$. The point
$e$ (resp. $f$) is the point of $[a,c]$ (resp. $[d,b]$) with an endpoint in
$F_1$, and such that there is no endpoint of $F_1$ after $e$ (resp. before $f$)
in $[a,c]$ (resp. $[d,b]$). Let $g$ (resp. $h$) be the first endpoint of $F$
after $e$ in $[a,c]$ (resp. before $f$ in $[d,b]$).

Let us denote by $F_2$ the set $F \setminus F_1$ (see Fig. \ref{progdyna2} for
an example). Let us prove that $F_1$ and $F_2$ are valid for $ab-ef$ and for
$gh-cd$, respectively. Since $F$ is $(ab-cd)$-dominating, all the chords with
both endpoints either in $[a,e]$ or in $[f,b]$ (resp. $[g,c]$ or $[h,d]$) are
dominated by $F$, thus by $F_1$ (resp. $F_2$). Therefore $F_1$ (resp. $F_2$) is
$(ab-ef)$-dominating (resp. $(gh-cd)$-dominating). Thus, by induction
hypothesis we have $v_{ac,ef}^1=v_{ac,ef}^f$ and $v_{gh,cd}^1=v_{gh,cd}^f$. And
since Property~\textbf{T1} can be applied for $F_1 \cup F_2$, by the safeness
of Property~\textbf{T1}, in Algorithm~\ref{algTree} we have $v_{ab,cd}^1 \leq
v_{ab,cd}^f$.
\end{proof}

Claims~\ref{cl:oneSenseOfInequality} and~\ref{cl:ineqTrees} together ensure
that $v_{ab,cd}^f = v_{ab,cd}^1$ and that $v_{ab,cd}^t = v_{ab,cd}^2$. Hence,
by dynamic programming all the regions of a given size can be found in
polynomial time. Let us now explain how we can verify if there is a dominating
set isomorphic to some tree of a given size $k$. This in particular will prove
Theorem~\ref{treepoly}.

\begin{claimN}
\label{claim:lastClaimTrees} Let $k$ be a positive integer. There is a
dominating tree of size at most $k$ in $C$ if and only if there is a region
$ab-cd$ such that $v^t_{ab,cd}\leq k$ and such that there is no chord strictly
contained in $[b,a]$ nor in $[c,d]$.
\end{claimN}
\begin{proof}
Assume that there is a region $ab-cd$ with $v^t_{ab,cd}\leq k$ and no chord
strictly contained in $[b,a]$ nor in $[c,d]$. Then by Claim~\ref{oneeachside},
all the chords with one endpoint in $[b,a]$ and one endpoint in $[a,b]$ are
dominated, and the same holds for the couple $c,d$. All the chords with both
endpoints in $[a,c] \cup [d,b]$ are dominated by Claim~\ref{extraforTree}, as
$T$ is valid for $ab-cd$. Since by assumption there is no chord strictly
contained in $[b,a]$ and in $[c,d]$, all the chords of the circle graph $C$ are
dominated, as by assumption there is no chord in the other intervals. Thus, $T$
is a dominating set.

Conversely, let $T$ be a dominating tree of size $k$. Let $uv$ be a chord of
$T$ which disconnects $T$. Thus $T \backslash \{uv\}$ has at least two
connected components $F_1$ and $F_2$. Let $a,c$ be the two extremities of the
first one,  and let $b,d$ be the extremities of the other ones (see
Fig.~\ref{progdyna1} for an illustration). Let us prove that $v_{ab,cd}^t \leq
k$. Indeed, $T$ is a tree by assumption, by definition it spans $a,b,c,d$, it
spans $ab-cd$ since the chord $uv$ is the unique chord from $[a,c]$ to $[d,b]$,
and it is $(ab-cd)$-dominating since $T$ is a dominating tree. In addition,
since $T$ is a dominating tree, there is no chord with both endpoints either in
the interval $[b,a]$ or in $[c,d]$, as otherwise such a chord would not be
dominated by $T$.
\end{proof}

By dynamic programming, Algorithm~\ref{algTree} computes in polynomial time the
regions for which there is a valid tree of any size from $1$ to $n$. Given a
region $ab-cd$ with $v_{ab,cd}^t \leq k$, we just have to verify that there are
no chords in the intervals $[b,a]$ and $[c,d]$, which can clearly be done in
polynomial time. One can easily check that Algorithm~\ref{algTree} runs in time
$\mathcal{O}(n^{10})$, but we did not make any attempt to improve its time
complexity. This completes the proof of Theorem~\ref{treepoly}.
\end{proof}

As Algorithm~\ref{algTree} computes the regions for which there is a valid tree
of any size from $1$ to $n$, it can be slightly modified to obtain the
following corollary.


\begin{corollary}
\label{cor:treeFixedSize} Let $\mathcal{T}_k$ be the set of all trees of size
exactly $k$. Then $\mathcal{T}_k$-\textsc{Dominating Set} can be solved in
polynomial time in circle graphs.
\end{corollary}

\subsection{FPT algorithm for a given tree}
\label{sec:FPTalgoTrees}

It turns out that when we seek a dominating set isomorphic to a given {\sl
fixed} tree $T$, the problem is \textsc{FPT} parameterized by $|V(T)|$. In
order to express the running time of our algorithm, and to prove that it is
{\sl subexponential} in $|V(T)|$, we need some definitions. Let $T$ be a tree,
and let us root $T$ at an arbitrary vertex $r$. Let $v$ be a vertex of $T$. We
denote by $T[v]$ the subtree of $T$ induced by $v$ and the descendants of $v$
in the rooted tree.  Let $v_1,\ldots,v_l$ be the children of $v$ in the tree
$T$ rooted at $r$. We define $F(v)$ as the forest $T[v_1] \cup T[v_2] \ldots
\cup T[v_l]$, which we consider as a multiset with elements $T[v_1], \ldots
T[v_l]$, where we consider two isomorphic trees $T[v_i]$ and $T[v_j]$ as the
same element.
Suppose that $F(v)$ contains exactly $s$ non-isomorphic trees $T_1,\ldots,T_s$,
and that for $1 \leq i \leq s$, there are exactly $d_i$ trees in $F(v)$ which
are isomorphic to $T_i$ (note that $\sum_{i=1}^{s}d_i = l$). We define the
following parameter, which corresponds to the number of non-isomorphic
sub(multi)sets of the multiset $F(v)$:
$$
\alpha_r^T(v)\ = \ \prod_{i=1}^{s} (d_i+1).
$$
Finally, we also define the following three parameters:
\begin{eqnarray*}
\alpha_r^T& = &\max_{v \in V(T)} \alpha_r^T(v)\\
\alpha^T & = & \max_{r {\scriptsize\mbox{ root of }}T} \alpha_r^T\\
\alpha^t & = & \max_{T: |V(T)|=t} \alpha^T.
\end{eqnarray*}

Let $t = |V(T)|$. Note that for any tree $T$, we easily have that $\alpha^T
\leq 2^t$, and that if $T$ has maximum degree at most $\Delta$, then it holds
that $\alpha^T \leq t \cdot 2^{\Delta -1}$ (by choosing $r$ to be a leaf of
$T$). In particular, if $T$ is a path on $t$ vertices, is holds that $\alpha^T
\leq 2t$. In the following proposition we upper-bound the parameter $\alpha^t$,
seen as a function of $t$, which will allow us to prove that the running time
of the algorithm in Theorem~\ref{th:FPTfixedTree} is subexponential.

\begin{proposition}\label{prop:boundALPHA}
$\alpha^t \ = \ 2^{\mathcal{O}(t \cdot \frac{\log \log t}{\log t})} \ = \
2^{o(t)}$.
\end{proposition}
\begin{proof} Let $t$ be an integer and let $T$ be a tree which maximizes $\alpha^t$, i.e., a
tree $T$ for which $\alpha^T = \alpha ^t$. Let $r$ be a root of $T$ maximizing
$\alpha_r^T$, and let $v$ be a vertex of $T$ such that $v$ maximizes $\alpha_r
^T(v)$. We claim that we can assume that $v=r$. Note that if $v \neq r$, then
it holds that $\alpha_v^T(v) > \alpha_r^T(v)$, as on the left-hand side we have
one more child of $v$ that contributes to $\alpha_v^T(v)$. Assume for
contradiction that $\alpha_r^T(v) > \alpha_r^T(r)$. Then, by the previous
inequality it holds that $\alpha_v^T(v) > \alpha_r^T(v) > \alpha_r^T(r)$,
contradicting the choice of $r$. Therefore, we assume henceforth that $v=r$.

Let $T_1,T_2,...,T_s$ be all the non-isomorphic rooted trees of $F(r)$ sorted
by increasing size (where size means number of vertices). As defined before, we
denote by $d_i$ the number of occurrences of $T_i$ in $F(r)$. By simplicity in
the sequel, let us denote by $k+1$ the size of $T_s$. We first want to find an
upper bound on $s$.

We will need the fact that the number of unlabeled rooted trees of size $\ell$
is asymptotically equal to $a \cdot d^{\ell} {\ell}^{-3/2}$, where $a \simeq
0.4399$ and $d \simeq 2.9558$ (see for instance~\cite[Chapter VII.5]{FlPh09}).
It follows that there exist two constants $c,c'$ and a constant $d$ such that
for all $\ell$, the number $N_t(\ell)$ of unlabeled rooted trees of size $\ell$
satisfies
\begin{equation}
\label{numbtrees} c \cdot d^{\ell} {\ell}^{-3/2}\ \leq\ N_t(\ell)\ \leq\ c'
\cdot d^{\ell} {\ell}^{-3/2}.
\end{equation}

\begin{claimN}\label{maxnumbtrees}
There exists a constant $c$ such that $s \leq c t/ \log t$.
\end{claimN}
\begin{proof}
Let us first prove by contradiction that all the trees of size at most $k$
appear in $F(v)$. Assume that there is a tree $T_a$ of size at most $k$ which
is not in $F(v)$. Let $T'$ be the same tree as $T$, rooted at $r'$, except that
we replace all the occurrences of $T_s$ in $F(r')$ by occurrences of $T_a$.
Since the size of $T_a$ is at most $k$, $T'$ contains strictly less vertices,
say $l$, than $T$. Thus in order to have $t$ vertices, we attach to the root of
$T'$ $l$ new trees isomorphic to the singleton-tree. Note that $T'$ has also
size $t$. Let us calculate the difference between $\alpha^T_r(r)$ and
$\alpha^{T'}_{r'}(r')$. Note that by maximality of $T$, we have $\alpha^T_r(r)
\geq \alpha^{T'}_{r'}(r')$. For all $1\leq i \leq s-1$ such that $T_i$ is not
the singleton-tree, by construction the number of trees isomorphic to $T_i$ in
$F(r)$ in $T$ is equal to the number of occurrences of $T_i$ in $F(r')$ in
$T'$. Thus $d_i =d_i'$ for all $1\leq i \leq s-1$ when $T_i$ is not the
singleton-tree. Since, by construction, there are the same number of
occurrences of $T_s$ in $F(r)$ and of $T_a$ in $F(r')$, we have $d_s=d_a$. Thus
the only difference between $\alpha_r^T(r)$ and $\alpha_{r'}^{T'}(r')$ is the
term corresponding to the singleton-tree. If the number of occurrences of the
singleton-tree in $F(r)$ is $d^*$, then the number of occurrences of the
singleton-tree in $F(r')$ is $d^*+l$. Thus $\alpha^{T'}_{r'}(r') >
\alpha^T_{r}(r)$, which contradicts the maximality of $T$.

Thus, we can assume that all the non-isomorphic rooted trees of size at most
$k$ appear in $F(r)$. By Equation~(\ref{numbtrees}), there are at least $c
\cdot d^k k^{-3/2}$ non-isomorphic rooted trees of size $k$. Therefore, if
there is a tree of size $k+1$ in $F(r)$, then necessarily
$$|V(T)|\ = \ t\ \geq\ k \cdot c' \cdot d^k k^{-3/2}\ =\ c \cdot d^k / \sqrt{k}.$$
Note that, in particular, we have that $k \geq \log t$ for $t$ large enough.
Indeed, when we replace $k$ by $\log t$, the inequality is satisfied. In the
following, when we write $\log$ we mean $\log_d$.

Since one can easily check that the number of rooted trees of size at most
$k+1$ is less that the number of rooted trees of size exactly $k+2$, and since
$T_s$ has size $k+1$, by the previous inequalities we have
$$ s\ \leq\  c \cdot d^{k+2} \cdot (k+2)^{-3/2}\ \leq\ c \cdot d^2 \cdot d^k /k^{3/2}
\ \leq\ c_1 / k \cdot d^k/\sqrt{k}\ \leq\ c_2 t/k\ \leq\  c_2 t/ \log t,$$
where the last inequality is a consequence of the fact that $k \geq  \log t$.
Thus, there exists a constant, called again $c$ for simplicity, such that $s
\leq c t / \log t$, as we wanted to prove.
\end{proof}

We now state a useful claim.

\begin{claimN}\label{maxiproduct}
Let $x_1,\ldots,x_k$ be some real variables and let $P$ be the polynomial such
that $P(x_1,\ldots,x_k)= \prod_{i=1}^k x_i$. Under the constraint $\sum_{i=1}^k
x_i= \ell$, the polynomial is maximized when $x_i=\ell/k$ for all $1 \leq i
\leq k$.
\end{claimN}
\begin{proof}
Assume for contradiction that this is not the case. Then by symmetry we can
assume that $x_1 > \ell/k$. Thus there exists another value, say $x_2$, such
that $x_2 < \ell/k$. Let $\varepsilon = \min\{x_1 - \ell/k, \ell/k - x_2\}$. To
compare the two values of the polynomial, we just have to compare the product
$x_1 \cdot x_2$. One can easily verify that $x_1 \cdot x_2 < (x_1 -
\varepsilon)(x_2 + \varepsilon)$, contradicting the fact that the polynomial
was maximized.\end{proof}

Claim~\ref{maxiproduct} ensures that $\alpha^T$ is maximized when all the
$d_i$'s are equal. Since $T$ contains $t$ vertices and since each tree contains
at least one vertex, it holds that $\sum_{i=1}^s d_i \leq t$. Thus the function
$\alpha^T$ is maximized when we have $d_i = t/s$ for all $1 \leq i \leq k$,
i.e.,
$$\alpha^T\ \leq\ \max_{ 1 \leq s \leq ct/\log t}\prod_{i=1}^{s}(d_i+1)\ \leq\ \max_{ 1 \leq s \leq ct/\log t}{(t/s+1)^s}.$$

\begin{claimN}\label{incrfunct}
Let $t$ be a large enough integer. The real function $f_t:x \mapsto (t/x+1)^x$
is increasing in the interval $[1,c t/\log t]$.
\end{claimN}
\begin{proof}
The derivative of the function $f_t$ is the following
$$f'_t(x)= \exp(x \log(t/x+1)) \cdot (\log(t/x+1)-1/(1+t/x)).$$
Note that the first term is always a positive function. We also have $1/(1+t/x)
\leq 1$. Thus we have $\log(t/x+1)-1/(1+t/x) \geq \log(t/x+1) -1 \geq 0$ when
$t$ is large enough, since $x \in [1,ct/\log(t)]$.
\end{proof}

Since $f_t$ is an increasing function by Claim~\ref{incrfunct}, and since $s
\leq c t/ \log t$ by Claim~\ref{maxnumbtrees}, we have that
$$\alpha^T\ \leq\ (\log t/c+1)^{c t/\log t}\ \leq\ 2^{c' \cdot t \cdot \frac{\log \log t}{\log t}}\ ,$$
for some constant $c'$, which completes the proof of
Proposition~\ref{prop:boundALPHA}.\end{proof}

We are now ready to state Theorem~\ref{th:FPTfixedTree}, which should be
compared to Theorem~\ref{thm:treeHARD} in Section~\ref{sec:hardTrees}. We use
Proposition~\ref{prop:boundALPHA} to conclude that the running time is
subexponential.


\begin{theorem}
\label{th:FPTfixedTree} Let $T$ be a given tree. There exists an \textsc{FPT}
algorithm to solve $\{ T \}$-\textsc{Dominating Set} in a circle graph on $n$
vertices, when parameterized by $t = |V(T)|$, running in time
$\mathcal{O}(\alpha^T \cdot n^{\mathcal{O}(1)}) = 2^{o(t)} \cdot
n^{\mathcal{O}(1)}$. In particular, if $T$ has bounded degree, $\{ T
\}$-\textsc{Dominating Set} can be solved in polynomial time in circle graphs.
\end{theorem}
\begin{proof}
The idea of the proof is basically the same as in the proof of
Theorem~\ref{treepoly}. The main difference is that in the proof of
Theorem~\ref{treepoly}, when Properties~\textbf{T1} or~\textbf{T2} are
satisfied, we can directly apply them and still obtain a forest or a tree. In
the current proof, when we make the union of two forests, we have to make sure
that the union of the two forests is still a subforest of $T$, and that we can
correctly complete it to obtain the desired tree $T$. For obtaining that, we
will apply the two properties stated below, whenever it is possible to create
forests which are induced by the children of the same vertex of $T$. Let us
first give some intuition on the algorithm.


In the following we consider the tree $T$ rooted at an arbitrary vertex $r$.
Let $w_1,\ldots,w_l$ be some vertices of $T$ which are children of the same
vertex $y$. The \emph{subforest of $T$ induced by $w_1,\ldots,w_l$}, denoted by
$F(w_1,\ldots,w_l)$, is the forest $T[w_1] \cup T[w_2] \ldots \cup T[w_l]$.


Roughly speaking, the idea of the algorithm is to exhaustively seek, for each
region $ab-cd$ and any possible subforest $F$ of $F(v)$ for every vertex $v$ in
$T$, a valid forest for $ab-cd$ isomorphic to $F$, and then try to grow it
until hopefully obtaining the target tree $T$. Note that if a vertex $v$ of $T$
has $k$ children, there are a priori $2^k$ possible subsets of children of $v$,
which define $2^k$ possible types of subforests in $F(v)$. But the key point is
that if some of the trees in $F(v)$ are isomorphic, some of the choices of
subsets of subforests will give rise to the same tree. In order to avoid this
redundancy, for each vertex $v$ of $T$, we partition the trees in $F(v)$ into
isomorphism classes, and then the choices within each isomorphism class reduce
to choosing the multiplicity of this tree, which corresponds to the parameter
$d_i+1$ (as we may not choose any copy of it) defined before the statement of
Proposition~\ref{prop:boundALPHA}. Note that carrying out this partition into
isomorphism classes can be done in polynomial time (in $t$) for each vertex of
$T$, using the fact that one can test whether two rooted trees $T_1$ and $T_2$
with $t$ vertices are isomorphic in $\mathcal{O}(t)$ time~\cite{AHU74}.

Therefore, if we proceed in this way, the number of such subforests for each
vertex $v \in V(T)$ is at most $\alpha_r^T(v)$. As we repeat this procedure for
every node of $T$, the cost of this routine per vertex is at most $\alpha_r^T =
\max_{\{v \in V(T)\}} \alpha_r^T(v)$. And as we chose the root arbitrarily, it
follows that the function can be upper-bounded by $\alpha^T = \max_{\{r
{\scriptsize\mbox{ root of $T$}\}}} \alpha_r^T$, which in turn can be
upper-bounded by $\alpha^t  = \max_{\{T: |V(T)|=t\}} \alpha^T$, which is a
subexponential function by Proposition~\ref{prop:boundALPHA}. We would like to
note that this step is the unique non-polynomial part of the algorithm.

Let us now explain more precisely the outline of the algorithm. An induced
subtree $T_1$ of the input circle graph is \emph{valid} for a region $ab-cd$
and a tree $T[w]$, if it is valid for $ab-cd$, and if, in addition, there is an
isomorphism between $T_1$ and $T[w]$ for which the unique chord between $[a,c]$
and $[b,d]$ corresponds to $w$. A forest $F_1$ is \emph{valid} for a region
$ab-cd$ and $F(w_1,\ldots,w_l)$, if it is valid for $ab-cd$, and if there is an
isomorphism between $F_1$ and $F(w_1,\ldots,w_l)$ for which the unique chord
between $[a,c]$ and $[b,d]$ of each connected component $T[w_i]$ corresponds to
vertex $w_i$. Let us now state the two properties that correspond to
Properties~\textbf{T1} and~\textbf{T2} of Theorem~\ref{treepoly}.



\begin{enumerate}

\item[\textbf{F1}]\label{r1fpt}
Let $F_1$ and $F_2$ be two valid forests for $ab-cd$ and $F(v_1,\ldots,v_l)$,
and for $ef-gh$ and $F(w_1, \ldots, w_m)$, respectively. Assume in addition
that $a \leq c \leq e \leq g \leq h \leq f \leq d \leq b$. Assume also that,
for all $1 \leq i \leq l$, $1 \leq j \leq m$, the vertices $v_i$ and $w_j$ are
pairwise distinct and are children of the same vertex $y$ of $T$. If there is
no chord with both endpoints either in $[c,e]$ or in $[f,d]$, then $F_1 \cup
F_2$ is valid for $ab-gh$ and $F(v_1,\ldots,v_l,w_1,\ldots,w_m)$.


\item[\textbf{F2}]\label{r2fpt} Let $F_1$ and $F_2$ be two
valid forests for $ab-cd$ and $F(v_1,\ldots,v_l)$, and for $ef-gh$ and
$F(w_1,\ldots,w_m)$, respectively ($F_2$ being possibly empty), and let $uv$ be
a chord of the input graph $C$. Assume that $u \leq a \leq c \leq e \leq g \leq
v \leq h \leq f \leq d \leq b$ and that there is no chord with both endpoints
either in $[u,a]$, or in $[g,v]$, or in $[v,h]$, or in $[b,u]$. Assume also
that there exists a vertex $y$ of $T$ with exactly $l+m$ children
$v_1,\ldots,v_l,w_1,\ldots,w_m$
Then $F_1 \cup F_2 \cup \{uv\}$ is a tree which is valid for $df-ce$ and
$T[y]$. When $F_2$ is empty, we consider that $e,f,g,h$ correspond to the point
$v$.
\end{enumerate}

In the proof of Theorem \ref{treepoly}, we have seen that the validity of the
corresponding regions is satisfied. Thus, we just have to verify that the tree
or the forest which is created is isomorphic to the target tree $T$, and that
the chords with one endpoint in each side are children of the same vertex. The
union of the two isomorphisms, and the fact that the chords with one endpoint
in both sides are the children of $y$, ensures that both properties are true.
Indeed, for example for Property~\textbf{F2}, since the chords with one
endpoint in each interval are exactly the children of $y$, it holds that the
chord corresponding to the vertex $y$ intersects exactly its children.

For each region $ab-cd$ and each tree $T[w]$, we define a boolean variable
$b_{ab,cd,w}^t$, which is set to `true' if and only if there is a valid tree
for $ab-cd$ and $T[w]$. For each region $ab-cd$ and each forest
$F(w_1,\ldots,w_l)$, we define a boolean variable $b_{ab,cd,w_1,\ldots,w_l}^f$
which is set to true if and only if there is a valid forest for $ab-cd$ and
$F(w_1,\ldots,w_l)$. (For the sake of simplicity, we distinguish between trees
and forests, but we would like to stress that it is not strictly necessary for
the algorithm.)

By a dynamic programming similar to Algorithm~\ref{algTree} in Theorem
\ref{treepoly}, we can compute all the regions $ab-cd$ and all vertices $v$ of
$T$ for which $b_{ab,cd,v}^t=\mbox{true}$ (and the same for forests). If there
is a region $ab-cd$ for which $b_{ab,cd,r}^t=\mbox{true}$, and such that there
is no chord with both endpoints either in $[b,a]$ or in $[c,d]$, then the tree
$T$ dominates all the chords in the input circle graph $C$. Indeed, the
safeness of Properties~\textbf{F1} and~\textbf{F2} ensures that there is a
valid tree isomorphic to $T$ for the region $ab-cd$. And
Claim~\ref{claim:lastClaimTrees} in the proof of Theorem~\ref{treepoly} ensures
that this tree is indeed a dominating tree.

Note that, indeed, the unique non-polynomial step of the algorithm consists in
generating the collection of non-isomorphic subforests, which are at most
$\alpha^T$ many. Thus, the dynamic programming algorithm runs in time
$\mathcal{O}(\alpha^T \cdot n^{\mathcal{O}(1)})$. Again, we did not make any
effort to optimize the degree of the polynomial in the running time.
\end{proof}

\noindent \textbf{Acknowledgment}. We would like to thank Sylvain Guillemot for
stimulating discussions that motivated some of the research carried out in this
paper.






\bibliography{circle}
\bibliographystyle{abbrv}

\end{document}

%% file: innerchords.pstex_t
\begin{picture}(0,0)%
\includegraphics{innerchords.pstex}%
\end{picture}%
\setlength{\unitlength}{4144sp}%
\begingroup\makeatletter\ifx\SetFigFont\undefined%
\gdef\SetFigFont#1#2#3#4#5{%
  \reset@font\fontsize{#1}{#2pt}%
  \fontfamily{#3}\fontseries{#4}\fontshape{#5}%
  \selectfont}%
\fi\endgroup%
\begin{picture}(2727,1246)(976,-1205)
\put(991,-1096){\makebox(0,0)[lb]{\smash{{\SetFigFont{9}{10.8}{\familydefault}{\mddefault}{\updefault}{\color[rgb]{0,0,0}$c_{i,j}$}%
}}}}
\put(1216,-1141){\makebox(0,0)[lb]{\smash{{\SetFigFont{9}{10.8}{\rmdefault}{\mddefault}{\updefault}{\color[rgb]{0,0,0}$1$$\ $ $2$ $3$ $\ \ \ \ \ \ $  $n$ $c_{i,j}'$ $c_{i,j+1}1$ $2$ $\ 3$ $\ \  \ \ \ \  n$ $c_{i,j+1}'$}%
}}}}
\end{picture}%

%% file: treedom.pstex_t
\begin{picture}(0,0)%
\includegraphics{treedom.pstex}%
\end{picture}%
\setlength{\unitlength}{2901sp}%
\begingroup\makeatletter\ifx\SetFigFont\undefined%
\gdef\SetFigFont#1#2#3#4#5{%
  \reset@font\fontsize{#1}{#2pt}%
  \fontfamily{#3}\fontseries{#4}\fontshape{#5}%
  \selectfont}%
\fi\endgroup%
\begin{picture}(4890,2208)(121,-1615)
\put(271,-1051){\makebox(0,0)[lb]{\smash{{\SetFigFont{8}{9.6}{\rmdefault}{\mddefault}{\updefault}{\color[rgb]{0,0,0}$g$}%
}}}}
\put(3331,119){\makebox(0,0)[lb]{\smash{{\SetFigFont{8}{9.6}{\rmdefault}{\mddefault}{\updefault}{\color[rgb]{0,0,0}$a$}%
}}}}
\put(4816,164){\makebox(0,0)[lb]{\smash{{\SetFigFont{8}{9.6}{\rmdefault}{\mddefault}{\updefault}{\color[rgb]{0,0,0}$b$}%
}}}}
\put(3151,-151){\makebox(0,0)[lb]{\smash{{\SetFigFont{8}{9.6}{\rmdefault}{\mddefault}{\updefault}{\color[rgb]{0,0,0}$c$}%
}}}}
\put(4996,-151){\makebox(0,0)[lb]{\smash{{\SetFigFont{8}{9.6}{\rmdefault}{\mddefault}{\updefault}{\color[rgb]{0,0,0}$d$}%
}}}}
\put(3286,-1141){\makebox(0,0)[lb]{\smash{{\SetFigFont{8}{9.6}{\rmdefault}{\mddefault}{\updefault}{\color[rgb]{0,0,0}$e$}%
}}}}
\put(4861,-1141){\makebox(0,0)[lb]{\smash{{\SetFigFont{8}{9.6}{\rmdefault}{\mddefault}{\updefault}{\color[rgb]{0,0,0}$f$}%
}}}}
\put(3556,-1411){\makebox(0,0)[lb]{\smash{{\SetFigFont{8}{9.6}{\rmdefault}{\mddefault}{\updefault}{\color[rgb]{0,0,0}$g$}%
}}}}
\put(4051,434){\makebox(0,0)[lb]{\smash{{\SetFigFont{8}{9.6}{\rmdefault}{\mddefault}{\updefault}{\color[rgb]{0,0,0}$u$}%
}}}}
\put(2082,-461){\makebox(0,0)[lb]{\smash{{\SetFigFont{8}{9.6}{\rmdefault}{\mddefault}{\updefault}{\color[rgb]{0,0,0}$d$}%
}}}}
\put(2064,-837){\makebox(0,0)[lb]{\smash{{\SetFigFont{8}{9.6}{\rmdefault}{\mddefault}{\updefault}{\color[rgb]{0,0,0}$f$}%
}}}}
\put(1931,-1139){\makebox(0,0)[lb]{\smash{{\SetFigFont{8}{9.6}{\rmdefault}{\mddefault}{\updefault}{\color[rgb]{0,0,0}$h$}%
}}}}
\put(4652,-1387){\makebox(0,0)[lb]{\smash{{\SetFigFont{8}{9.6}{\rmdefault}{\mddefault}{\updefault}{\color[rgb]{0,0,0}$h$}%
}}}}
\put(316,-16){\makebox(0,0)[lb]{\smash{{\SetFigFont{8}{9.6}{\rmdefault}{\mddefault}{\updefault}{\color[rgb]{0,0,0}$a$}%
}}}}
\put(136,-376){\makebox(0,0)[lb]{\smash{{\SetFigFont{8}{9.6}{\rmdefault}{\mddefault}{\updefault}{\color[rgb]{0,0,0}$c$}%
}}}}
\put(181,-826){\makebox(0,0)[lb]{\smash{{\SetFigFont{8}{9.6}{\rmdefault}{\mddefault}{\updefault}{\color[rgb]{0,0,0}$e$}%
}}}}
\put(1936,-16){\makebox(0,0)[lb]{\smash{{\SetFigFont{8}{9.6}{\rmdefault}{\mddefault}{\updefault}{\color[rgb]{0,0,0}$b$}%
}}}}
\put(4006,-1546){\makebox(0,0)[lb]{\smash{{\SetFigFont{8}{9.6}{\rmdefault}{\mddefault}{\updefault}{\color[rgb]{0,0,0}$v$}%
}}}}
\end{picture}%

%% file: progdyna1.pstex_t
\begin{picture}(0,0)%
\includegraphics{progdyna1.pstex}%
\end{picture}%
\setlength{\unitlength}{2693sp}%
\begingroup\makeatletter\ifx\SetFigFont\undefined%
\gdef\SetFigFont#1#2#3#4#5{%
  \reset@font\fontsize{#1}{#2pt}%
  \fontfamily{#3}\fontseries{#4}\fontshape{#5}%
  \selectfont}%
\fi\endgroup%
\begin{picture}(5788,2192)(796,-2650)
\put(811,-1006){\makebox(0,0)[lb]{\smash{{\SetFigFont{8}{9.6}{\rmdefault}{\mddefault}{\updefault}$a$}}}}
\put(4726,-2536){\makebox(0,0)[lb]{\smash{{\SetFigFont{8}{9.6}{\rmdefault}{\mddefault}{\updefault}{\color[rgb]{0,0,0}$c$}%
}}}}
\put(4321,-2086){\makebox(0,0)[lb]{\smash{{\SetFigFont{8}{9.6}{\rmdefault}{\mddefault}{\updefault}{\color[rgb]{0,0,0}$f$}%
}}}}
\put(4276,-1186){\makebox(0,0)[lb]{\smash{{\SetFigFont{8}{9.6}{\rmdefault}{\mddefault}{\updefault}{\color[rgb]{0,0,0}$e$}%
}}}}
\put(6346,-826){\makebox(0,0)[lb]{\smash{{\SetFigFont{8}{9.6}{\rmdefault}{\mddefault}{\updefault}{\color[rgb]{0,0,0}$b$}%
}}}}
\put(4210,-1606){\makebox(0,0)[lb]{\smash{{\SetFigFont{8}{9.6}{\rmdefault}{\mddefault}{\updefault}{\color[rgb]{0,0,0}$u$}%
}}}}
\put(6554,-1843){\makebox(0,0)[lb]{\smash{{\SetFigFont{8}{9.6}{\rmdefault}{\mddefault}{\updefault}{\color[rgb]{0,0,0}$v$}%
}}}}
\put(6569,-1468){\makebox(0,0)[lb]{\smash{{\SetFigFont{8}{9.6}{\rmdefault}{\mddefault}{\updefault}{\color[rgb]{0,0,0}$g$}%
}}}}
\put(6357,-2209){\makebox(0,0)[lb]{\smash{{\SetFigFont{8}{9.6}{\rmdefault}{\mddefault}{\updefault}{\color[rgb]{0,0,0}$h$}%
}}}}
\put(4456,-826){\makebox(0,0)[lb]{\smash{{\SetFigFont{8}{9.6}{\rmdefault}{\mddefault}{\updefault}{\color[rgb]{0,0,0}$a$}%
}}}}
\put(6076,-2581){\makebox(0,0)[lb]{\smash{{\SetFigFont{8}{9.6}{\rmdefault}{\mddefault}{\updefault}{\color[rgb]{0,0,0}$d$}%
}}}}
\put(1130,-2460){\makebox(0,0)[lb]{\smash{{\SetFigFont{8}{9.6}{\rmdefault}{\mddefault}{\updefault}{\color[rgb]{0,0,0}$c$}%
}}}}
\put(2917,-1022){\makebox(0,0)[lb]{\smash{{\SetFigFont{8}{9.6}{\rmdefault}{\mddefault}{\updefault}{\color[rgb]{0,0,0}$b$}%
}}}}
\put(2701,-2446){\makebox(0,0)[lb]{\smash{{\SetFigFont{8}{9.6}{\rmdefault}{\mddefault}{\updefault}{\color[rgb]{0,0,0}$d$}%
}}}}
\end{picture}%

%% file: progdyna2.pstex_t
\begin{picture}(0,0)%
\includegraphics{progdyna2.pstex}%
\end{picture}%
\setlength{\unitlength}{2901sp}%
\begingroup\makeatletter\ifx\SetFigFont\undefined%
\gdef\SetFigFont#1#2#3#4#5{%
  \reset@font\fontsize{#1}{#2pt}%
  \fontfamily{#3}\fontseries{#4}\fontshape{#5}%
  \selectfont}%
\fi\endgroup%
\begin{picture}(4665,1826)(76,-1469)
\put(1891,-1186){\makebox(0,0)[lb]{\smash{{\SetFigFont{8}{9.6}{\rmdefault}{\mddefault}{\updefault}{\color[rgb]{0,0,0}$d$}%
}}}}
\put( 91,-1006){\makebox(0,0)[lb]{\smash{{\SetFigFont{8}{9.6}{\rmdefault}{\mddefault}{\updefault}{\color[rgb]{0,0,0}$c$}%
}}}}
\put(2746,-286){\makebox(0,0)[lb]{\smash{{\SetFigFont{8}{9.6}{\rmdefault}{\mddefault}{\updefault}{\color[rgb]{0,0,0}$a$}%
}}}}
\put(4726,-376){\makebox(0,0)[lb]{\smash{{\SetFigFont{8}{9.6}{\rmdefault}{\mddefault}{\updefault}{\color[rgb]{0,0,0}$b$}%
}}}}
\put(2701,-556){\makebox(0,0)[lb]{\smash{{\SetFigFont{8}{9.6}{\rmdefault}{\mddefault}{\updefault}{\color[rgb]{0,0,0}$e$}%
}}}}
\put(4726,-646){\makebox(0,0)[lb]{\smash{{\SetFigFont{8}{9.6}{\rmdefault}{\mddefault}{\updefault}{\color[rgb]{0,0,0}$f$}%
}}}}
\put(2746,-826){\makebox(0,0)[lb]{\smash{{\SetFigFont{8}{9.6}{\rmdefault}{\mddefault}{\updefault}{\color[rgb]{0,0,0}$g$}%
}}}}
\put(4456,-1366){\makebox(0,0)[lb]{\smash{{\SetFigFont{8}{9.6}{\rmdefault}{\mddefault}{\updefault}{\color[rgb]{0,0,0}$d$}%
}}}}
\put(4636,-1141){\makebox(0,0)[lb]{\smash{{\SetFigFont{8}{9.6}{\rmdefault}{\mddefault}{\updefault}{\color[rgb]{0,0,0}$h$}%
}}}}
\put( 91,-151){\makebox(0,0)[lb]{\smash{{\SetFigFont{8}{9.6}{\rmdefault}{\mddefault}{\updefault}{\color[rgb]{0,0,0}$a$}%
}}}}
\put(2026,-331){\makebox(0,0)[lb]{\smash{{\SetFigFont{8}{9.6}{\rmdefault}{\mddefault}{\updefault}{\color[rgb]{0,0,0}$b$}%
}}}}
\put(2836,-1096){\makebox(0,0)[lb]{\smash{{\SetFigFont{8}{9.6}{\rmdefault}{\mddefault}{\updefault}{\color[rgb]{0,0,0}$c$}%
}}}}
\end{picture}%